\documentclass[10pt]{article}

\usepackage{./kashsty}
\usepackage{./kashthm}
\usepackage{./kashdef}
\usepackage{./moredef}

\usepackage[round,authoryear]{natbib}
\usepackage[margin=1.25in]{geometry}

\def\probB#1{\mathbb{P}\left(#1\right)}

\begin{document}
\title{Random Invariance Testing on Quadratic Form
Statistics with Application to Autocorrelation}
%\title{A sample article title with some additional note\thanksref{t1}}
%\runtitle{Random Invariance Testing on Quadratic Form
%tatistics}
%\thankstext{T1}{A sample additional note to the title.}

%%%%%%%%%%%%%%%%%%%%%%%%%%%%%%%%%%%%%%%%%%%%%%%
%% Only one address is permitted per author. %%
%% Only division, organization and e-mail is %%
%% included in the address.                  %%
%% Additional information can be included in %%
\author{
{Amitakshar}~{Biswas}\footnote{
Department of Statistics, University of Illinois Urbana-Champaign, Champaign, IL, 61820, USA
biswas8@illinois.edu}
 ~and
{Adam B}~{Kashlak}\footnote{
Department of Mathematical and Statistical Sciences, University of Alberta, Edmonton, AB, T6G 2G1, Canada
kashlak@ualberta.ca} \\
%\orcid{0000-0002-4050-7784}
}
%%%%%%%%%%%%%%%%%%%%%%%%%%%%%%%%%%%%%%%%%%%%%%
%% Addresses                                %%
%%%%%%%%%%%%%%%%%%%%%%%%%%%%%%%%%%%%%%%%%%%%%%

\maketitle

\begin{abstract}
Randomization testing with permutations is a very 
  common nonparametric approach to hypothesis testing.  
  However, 
  randomization testing can be done with other group
  transformations including random rotations.
  In this work, we consider the problem of invariance 
  in quadratic form statistics under a unified 
  framework with closed form p-values.  
  In particular, 
  we propose a nonparametric variant of the 
  classic Durbin-Watson test for testing for autocorrelation
  in time series data at arbitrary lags.  Our test is performed
  by integrating over a group of invariances of the test statistic,
  and easy-to-compute analytic formulae for the p-value 
  are derived from concentration inequalities on compact groups on a per-lag basis.  
  Thus, the usual necessity of large-scale Monte Carlo simulations
  is rendered unnecessary. 
  Our tests outperform the classic 
  Breusch-Godfrey and Ljung-Box tests on simulated data with respect 
  to statistical power to identify significant autocorrelation.
  They also can be used to identify the presence of autocorrelation
  at large lags such as in monthly solar intensity data, 
  which follows an approximate 11 year (132 month) cycle.
  The general formulation of this approach can be easily adapted
  to other quadratic form statistics.\\

\end{abstract}

%\begin{keyword}[class=MSC]
%\kwd[Primary ]{62G09}
%\kwd[; secondary ]{62M10}
%\end{keyword}

%\begin{keyword}
%\kwd{Breusch-Godfrey Test}
%\kwd{Concentration Inequalities}
%\kwd{Durbin-Watson Test}
%\kwd{Nonparametric Test}
%\kwd{Randomization Test}
%\end{keyword}

%\end{frontmatter}
%%%%%%%%%%%%%%%%%%%%%%%%%%%%%%%%%%%%%%%%%%%%%%
%% Please use \tableofcontents for articles %%
%% with 50 pages and more                   %%
%%%%%%%%%%%%%%%%%%%%%%%%%%%%%%%%%%%%%%%%%%%%%%
%\tableofcontents

\section{Introduction}

When testing the second order properties of a data set
or a statistical model such as variance and covariance,
statistical hypothesis tests often take the form of
a quadratic form $\TT{x}Ax$ for some vector $x\in\real^n$
and some matrix $A\in\real^{n\times n}$.
Classical theory of the 20th century invariably relied
on assuming the asymptotic distribution of a test statistic
if such a distribution were tractable.  More modern 
computational approaches have been considered such as 
Monte Carlo tests and permutation tests, if applicable,
that rely on massive computer simulation.
However, quadratic forms often yield a natural invariance
to rotation matrices $M$, which are those such that 
$M\TT{M} = \TT{M}M = I$ and $\det(M)=1$.
For example, such matrices $M$ preserve the 
Euclidean inner product $\iprod{x}{y} = \iprod{Mx}{My}$.
Thus, we propose a class of randomization tests 
\cite[Section 15.2]{LEHMANN2006} making use of 
$M$ uniformly distributed in the special orthogonal
group $SO(n)$ and leverage the concentration
of measure
properties of
Haar measure on $SO(n)$ to derive analytic 
formulae for such nonparametric statistical tests.
The tangible result of this work is a computationally-trivial
nonparametric test for autocorrelation in time series
data that outperforms existing tests;
see Equations~\ref{eqn:betaPV} and~\ref{final} below.
More generally, this framework can be easily extended more 
widely to other test statistics in the form of a
quadratic form.  This formulation provides a uniform
treatment of three main symmetry groups.
In the context of randomization tests, these three crucial
symmetries are reflectability, exchangeability, 
and rotateability 
\citep{KALLENBERG2006}. 

In $\real^n$, the number of reflections is $2^n$ corresponding
to the vertices of the hypercube, the number of permutations
is $n!$, and rotations from $SO(n)$ are a compact topological 
group.
It is important to note a clear difference
in the simulation process: Monte Carlo simulation of uniformly
random reflections or permutations is computationally trivial 
compared to
generating uniformly random rotation matrices 
especially with large datasets.  Standard algorithms
for Monte Carlo sampling of a uniform element of
$SO(n)$ require the generation of an $n\times n$
random matrix followed by a QR decomposition.
This in general makes p-values derived from 
such random rotation tests infeasible 
to compute with brute force as one would do 
when performing a permutation test
\citep{PESARIN2010,GOOD2013}.
However, more clever problem-specific approaches
can be applied to make some computation feasible.
The scientific literature on random rotation
tests is not as vast as that of permutation tests,
but has a few notable works
such as the foundational article of \cite{LANGSRUD2005} 
and the more recent \cite{SOLARI2014}, which 
proposes a rotation-based 
multiple testing procedure for use in microarray studies.  
See also the \texttt{flip} package in R 
\citep{FLIP}.

Detailed constructions of permutation and more general 
randomization tests are discussed in many
sources \citep{LEHMANN2006,PESARIN2010,GOOD2013}
and a recent result on the consistency of such 
tests can be found in \cite{DOBRIBAN2022}.
Let $Y\in\real^n$ be a random vector and 
let $\mathcal{O}$ be a group of transformations such that
%\begin{equation}\label{rdef}
$
  Y = OY 
$
in distribution under $H_0$
%\end{equation}
for every ${O} \in \mathcal{O}$. The main advantage of 
this condition is that it provides a collection
of observations
${O}Y$, which all have the same 
probability of occurring under $H_0$.
This is referred to as the randomization 
hypothesis and the transformations ${O}$ as null-invariants. 
They are called null-invariants as the transformation of 
$Y$ by ${O}$ has no effect on the underlying null distribution.

In this article, we propose an unified 
concentration-inequality-based treatment of randomization
tests for quadratic form statistics with an easy-to-compute
closed form test statistics and p-values.
In particular, we propose
Equation~\ref{final}
for testing for autocorrelation in time series
data at a specific lag $h$.  This formula is based on 
concentration inequalities for Haar measure
on $SO(n)$ from \cite{MECKES2019}.
In comparison with classic tests for autocorrelation,
the Durbin-Watson test is only implemented to 
test for
autocorrelation at lag $h=1$ whereas the 
Breusch-Godfrey and Ljung-Box tests are omnibus tests 
that can only reject the 
general null hypothesis of no autocorrelation
at any lag $1,\ldots,h$.
R code for performing all of the simulation
and real data experiments from 
Sections~\ref{sec:sims} and~\ref{sec:solarData}
can be found at 
\url{https://github.com/cachelack/Random-Rotation-Tests.git}.

\begin{remark}[Feasibility of Computation]
A quick timing test for each of the three randomizations
is displayed in Figure~\ref{fig:runtime} for 
dimensions 100, 200, and 400.  This test was performed in 
RStudio with an Intel(R) Core(TM) i7-8700K CPU at 3.70GHz.
Random reflections were generated via \texttt{rbinom()},
random permutations with \texttt{sample()}, and 
random rotations via \texttt{rnorm()} to create a 
random Gaussian matrix and then \texttt{qr()} for the 
QR-decomposition.
Recent work on computing very small tail probabilities
for quadratic forms of multivariate normal random 
vectors can be found in \cite{ONG2026}.
There are also modern approaches to efficient sampling 
from subgroups instead of exhausting
the whole group of randomization transformations
\citep{KONING_HEMERIK_2023,KONING2024}.

Of note, the timing comparison in Figure~\ref{fig:runtime}
pertains to generic Lipschitz statistics on $SO(n)$, for
which the full rotation matrix must be generated. The value of the
concentration approach developed in this work is that it yields
closed-form p-values that apply uniformly across  
symmetry groups and arbitrary quadratic forms, with no spectral
computation or numerical inversion.  See Section~\ref{sec:exactChiSquare}
below.
\end{remark}

\begin{figure}
    \centering
    \includegraphics[width=0.7\linewidth]{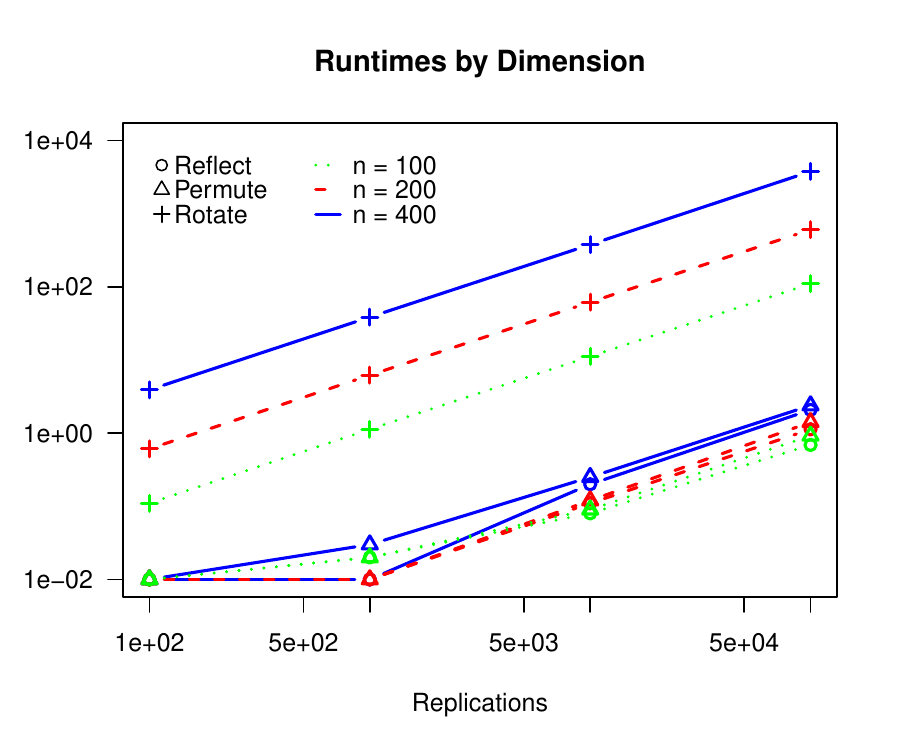}
    \caption{Computation 
      time of $100,\ldots,10000$ randomizations in 
      dimensions $100,200,400$.
    }
    \label{fig:runtime}
\end{figure}

\section{Quadratic Form Statistics}
\label{sec:quadForm}

\subsection{Invariance in Quadratic Forms}

In general randomization testing, we invoke
the so-called \textit{randomization hypothesis},
which is that the distribution of our 
random data vector $X\in\real^n$ is invariant
to any choice of group action.  That is, for 
a compact group $\mathcal{G}$, under the null hypothesis,
$X= gX$ in distribution.  In this article, we 
mainly focus on $g$ being rotation matrices from 
$SO(n)$.  However, we will also consider reflection
matrices, being diagonal matrices with diagonal entries
of $\pm1$, and permutation matrices, being binary matrices
with exactly one 1 in every row and column.  All
three of these collections are subgroups of the orthogonal
group $O(n)$.  We will not work with $O(n)$ in its totality 
as it has two disconnected components, which 
breaks our use of concentration of measure below.

A general quadratic form statistic takes the form
$q(x;A) = \TT{x}Ax$ for some fixed matrix $A\in\real^{n\times n}$.
We can, without loss of generality, assume that $A$ is
symmetric since $q(x;A) = q(x;\TT{A})$ allowing us to
replace $A$ with $(A + \TT{A})/2$.
In hypothesis testing with $X$ from the null distribution,
and $x$ a vector of observations, 
we will typically be concerned 
with 
$$
  \probB{ 
    \abs{q(X;A) - \xv q(X;A)} \ge  
    \abs{q(x;A) - \xv q(X;A)}   
  }
$$
being the probability of observing a more extreme
value than $q(x;A)$ with respect to the mean
$\xv q(X;A) = \sum_{i,j=1}^n a_{i,j}\cov{X_i}{X_j}$.

In terms of the second order properties of 
the random vector $X$, invariance under elements
of $SO(n)$ implies that $X$ is an isotropic random
vector, i.e. $\mathrm{cov}(X) = cI_n$ for some $c>0$.
Similarly, if $X$ is invariant under permutations 
(exchangeable), then the covariance matrix $X$ has
constant diagonal entries and a (different) constant 
for all off-diagonal entries.
Lastly, if $X$ has reflective symmetry, then its
covariance matrix 
can have arbitrary positive values on the diagonal, but must
have zeros on the off-diagonal.
An interesting question is to consider which group
of symmetries is most appropriate for a given statistical
test.  With respect to the covariance structure of the 
data, invariance to either rotations or permutations 
implies an homogeneous variance whereas invariance 
to either rotations or reflections implies 
zero correlation between entries in the vectors. 

In each of these settings, we can leverage 
concentration of Haar measure on the given 
group of transformations.  
A main result used in this work is a quick corollary 
to Theorem 5.17 in \cite{MECKES2019}.  The general 
theorem demonstrates sub-Gaussian concentration of
Haar measure on the groups 
$SO(n),\, SU(n),\, U(n), \text{ and } Sp(2n)$.
\begin{corollary} \label{meck}
    Let $\mathcal{G} = \mathrm{SO}(n)$ be equipped with the Hilbert-Schmidt metric. Suppose that $F : \mathcal{G} \to \real$ is L-Lipschitz and that $M \in \mathcal{G}$ is a Haar-distributed random matrix. Then for each $t > 0$,
    $
        \probB{F(M) \geq \mathrm{E}F(M) + t} \leq \exp{\{{-(n-2)t^2}/{16L^2}\}}.
    $
\end{corollary}
Any quadratic form on $SO(n)$ is a Lipschitz 
function by continuity and compactness.
In the following theorem, we compute the
Lipschitz constant for a randomly rotated
quadratic form on $SO(n)$ with respect
to the Hilbert-Schmidt metric. 

\begin{theorem}\label{thm:lipGen}
For a fixed $x \in S^{n-1}$, the function 
$F: SO(n) \rightarrow \real$ defined by 
$F(M) = \TT{(Mx)}A(Mx)$
is Lipschitz with Lipschitz constant $2\rho(B)$ 
with $\rho(B)$ the spectral radius of the 
symmetrized matrix $B = (A+\TT{A})/2$
where the group $\mathrm{SO}(n)$ is equipped with the 
Hilbert-Schmidt metric.

Consequently,
$$
    \probB{
      F(M) \geq \mathrm{E}F(M) + t
    } \leq \exp{\{{-(n-2)t^2}/{64\rho(B)^2}\}}.
$$
\end{theorem}

A similar result to concentration on the classical 
compact groups from \cite{MECKES2019} is concentration
on finite metric spaces
detailed in \cite{ledoux} Section 4.1.  However,
mapping the proof of the previous theorem to
quadratic forms on such finite groups does not yield 
bounds of the correct order.  
To illustrate this problem in the case of 
random reflections, consider that for any two 
diagonal matrices $R_1,R_2$ with diagonal entries
of $\pm1$,
$$
  \norm{R_1 - R_2}_\text{OP} = \left\{
    \begin{array}{ll}
      0 & \text{ if } R_1=R_2\\
      2 & \text{ if } R_1\ne R_2.
    \end{array}
  \right.
$$
Hence, we approach the finite groups differently.

%The following two theorems are included for the sake
%of comparison but not directly used in the subsequent
%anaalysis.
For reflective symmetry, dimension-free 
sub-Gaussian concentration is 
achieved through a quick application of 
a decoupling inequality from \cite{KWAPIEN1987} on 
polynomial, specifically Rademacher, chaoses.
\begin{theorem}
  \label{thm:reflectConc}
  Let $\mathcal{G} = \mathbb{B}(n) = \{-1,1\}^n$, fixed vector $x\in S^{n-1}$,
  random vector $b\in \mathbb{B}(n)$ with respect to uniform (Haar) measure,
  and $F:\mathbb{B}(n)\rightarrow\real$ defined by 
  $F(b) = \TT{(b\circ x)}A(b\circ x) = \sum_{i,j=1}^n a_{i,j}b_ib_jx_ix_j$.
  Then,
  $$
    \probB{
      F(b) - \xv F(b) > t
    } \le \exp\left(
      -t^2/ 64\max_{i>j}\{a_{i,j}^2\}
    \right)
  $$
\end{theorem}

Lastly, for exchangeability, we have a similar sub-Gaussian
bound to those above.  This one follows from application
of the Azuma-Hoeffding inequality after considering an
arbitrary
permutation $\pi$ as a product of transpositions. 
\begin{theorem}
  \label{thm:exchangeConc}
  Let $\mathcal{G} = \mathbb{S}(n)$, the symmetric group on 
  $n$ elements, $x\in S^{n-1}$ a fixed vector, 
  and $F:G\rightarrow\real$ defined by 
  $F(\pi) = \TT{(\pi x)}A(\pi x) = 
  \sum_{i,j=1}^n a_{i,j}x_{\pi(i)}x_{\pi(j)}$
  for $\pi$ a uniformly distributed 
  permutation of the set $\{1,\ldots,n\}$.
  Then, for all $t>0$,
  $$
     \probB{
       F(\pi) - \xv F(\pi) > t 
     } \le \exp\left(
       -t^2/16\rho(B)^2
     \right)
   $$
   where $\rho(B)$ is the spectral radius of
  $B = (A + A^{\mathrm T})/2$.
\end{theorem}

\subsection{Autocorrelation Statistics}

Test statistics expressed as quadratic forms are 
ubiquitous across the field of statistical hypothesis
testing.  Here, we consider such statistics for 
testing for autocorrelation in time series data.
A time series $\{X_t\}_{t=1}^n$ 
is called an autoregressive process of order $p$, 
denoted AR$(p)$ for short, if $X_t$ has zero mean, 
and we can write it as
$
    X_t = \sum_{i=1}^p \phi_i X_{t-i} + \varepsilon_t
$
where $\varepsilon_t$ is uncorrelated white noise, i.e.
$\xv(\veps_t)=0$ and 
$\var{\veps_t}=\sigma^2$ for all $t$ and
$\cov{\veps_t}{\veps_s}=0$ for all $s\ne t$, and 
where $\phi_i \in \mathbb{R}$ with $\phi_p\ne0$.
Typically, one considers such processes where the 
$\phi_i$ are chosen to make the process stationary 
and causal; see, for example, \cite{SHUMWAY2000}.

In particular, for the stationary AR$(1)$ process, the single 
parameter $\phi$ is the autocorrelation between 
$X_{t-1}$ and $X_t$ for all $t \in \{2,\dots,n\}$. 
We assume throughout that the mean of the time series
is zero.  In practice, we can always center the data 
by defining a new time series $X_t \leftarrow X_t - \bar{X}$ 
for all $t$. In that case, our analysis holds but in 
$n-1$ dimensions.

We denote the autocovariance at lag $h$
to be 
$K(h) = \cov{X_t}{X_{t+h}}$, which is assumed 
invariant to $t$ by stationarity of $X_t$.  
We further denote the autocorrelation to be 
$\eta(h) = K(h)/K(0)$ with $K(0)=\var{X_t}$.
The sample autocorrelation at lag 1 
is estimated by the statistic,
$
  \hat{\eta}_1 = 
  ({n\hat{\sigma}^2})^{-1}\sum_{i = 1}^{n-1} X_iX_{i+1}
$
where 
$
  \hat{\sigma}^2 = 
  {n}^{-1}\sum_{i = 1}^{n} X_i^2 = 
  {n}^{-1}\lVert X \rVert^2
$, 
is the sample variance and $X = \TT{(X_1, \dots, X_n)}$.
We can write the lag 1 autocorrelation estimator as,
$$
  \hat{\eta}_1 = \frac{\TT{X}AX}{\lVert X \rVert^2} = 
  \frac{\TT{X}}{\lVert X \rVert}A\frac{X}{\lVert X \rVert}
  ~\text{ where }~
  A_{n \times n} = 
  \begin{bmatrix}
    \boldsymbol{0}_{n-1,1} & I_{n-1} \\
    0 & \boldsymbol{0}_{1,n-1}
  \end{bmatrix}
$$
for $\boldsymbol{0}_{n-1,1} = \TT{\boldsymbol{0}_{1,n-1}}$
a vector of $n-1$ zeros.
Upon writing $\hat{\eta}_1$ in this way, we can, 
without loss of generality, assume that $X$ 
lies in the $n$-dimensional unit sphere, 
i.e., $X \in S^{n-1}$ and $\hat{\eta}_1 = \TT{X}AX$.
For any lag $h>0$, we similarly have that the
estimated autocorrelation is 
$\hat{\eta}_h = \TT{X}A^hX$.

There are many statistical methods to test for the 
presence of autocorrelation in time series data. 
One of the most well-known tests is the Durbin-Watson test
\citep{DurbinWatsonA,DurbinWatsonB,DurbinWatsonC}.
The Durbin-Watson test tests for
autocorrelations at lag $1$ among the residuals 
of a linear model.
Considering the linear model
$
  Y_t = \beta_0 + \beta_1t + \dots + \beta_pt^p + r_t
$
with $t = 1, \dots, n$, 
we can compute the least squares estimator 
$\hat{\beta}$ and then compute
the model residuals 
$\hat{r}_t = Y_t - \langle \hat{\beta}, (1,t, \dots, t^p)\rangle$. 
The Durbin–Watson test assumes the
following model for the residuals:
$
  \hat{r}_t= \phi \hat{r}_{t-1} + w_t
$
where $w_t$ is uncorrelated white noise. 
Then, it tests the hypotheses 
$H_0 : \phi = 0$ vs $H_1 : \phi \neq 0$. 
It does this by computing the test statistic,
$$
 Q_{DW} = 
 \frac{
   \sum_{t=2}^n (\hat{r}_t - \hat{r}_{t-1})^2
 }{
   \sum_{t=1}^n \hat{r}_t^2
 } \approx
 2 - 
 2\frac{\TT{\hat{r}}A\hat{r}}{\TT{\hat{r}}\hat{r}}
$$
for large $n$.
If this test statistic is close to zero, 
it implies that $\hat{r}_t$ and $\hat{r}_{t-1}$ 
are close in value
indicating a strong positive autocorrelation at lag $1$. 
In contrast, if the test statistic is large 
(close to the max of $4$), then it indicates that 
there is a strong negative
autocorrelation at lag $1$. 
Otherwise, a test statistic near $2$ indicates 
no autocorrelation of order $1$.
Significance testing is typically performed by 
assuming that the noise process $w_t$ is Gaussian
and approximating the null distribution of $Q_{DW}$ 
by a 
weighted sum of chi-squared random variables;
see, for example, \texttt{dwtest} in the 
\texttt{lmtest} package in R \citep{LMTEST} or 
alternatively 
\texttt{dw.test} in the \texttt{desk} package
\citep{DESK}, which applies Imhof's algorithm
\citep{IMHOF1961} among other methods.
There is much literature on approximations
to the null distribution of $Q_{DW}$
for various of statistical models
\citep{INDER1986,KING1991}.

The Durbin-Watson test is typically only 
implemented for testing for lag 1 autocorrelation.
We will demonstrate in the following section that
our random rotation test can trivially extend the
Durbin-Watson test to higher lags and without the
necessity of Gaussian errors.
An alternative test for autocorrelation at higher
lags is the Breusch–Godfrey test
\citep{BREUSCH1978,GODFREY1978}, which we will 
compare against in the subsequent sections.
The Breusch–Godfrey test computes the $R^2$ value for
the linear model
$$
  \hat{r}_t = 
  \beta_0 + \beta_1t + \ldots + \beta_pt^p + 
  \phi_1\hat{r}_{t-1} + \ldots + \phi_h\hat{r}_{t-h}
  + \veps_t
$$
and uses the fact that $nR^2$ is asymptotically
$\distChiSquared{h}$ under the null hypothesis 
that $\phi_1 = \ldots = \phi_h = 0$.  Thus, it is
a simultaneous test for presence of autocorrelation 
at any lag from 1 to a user specified $h$.  This
test is implemented in \texttt{bgtest} in 
the \texttt{lmtest} R package \citep{LMTEST}.

A third asymptotic test for autocorrelation is the 
Ljung–Box test, which similar to the above yields 
a chi-squared distribution asymptotically \citep{LJUNGBOX1978}.
The test statistic is 
$$
  Q_{LB} = n(n+2)\sum_{k=1}^h \frac{\hat\eta_k^2}{n-k}
$$
where $\hat{\eta}_h$ is the sample autocorrelation at lag $h$.
Here, $Q_{LB}$ is asymptotically $\distChiSquared{h}$ as
$n\rightarrow\infty$.  This test is implemented in R's base
\texttt{stats} package via the function \texttt{Box.test}.
This function also includes the precursor Box-Pierce 
test \citep{BOXPIERCE1970}.

In summary, Durbin-Watson assumes normal errors and
has a null distribution that is numerically computed
by one of a variety of algorithms.  
In contrast, Breusch-Godfrey and
Ljung-Box both rely on large samples and the central
limit theorem to asymptotically converge to a chi-squared
random variable under the null.

\section{Random Invariance  Testing}
\label{sec:randRot}

\subsection{Sub-Gaussian Bound}
\label{sec:subGauss}

Our random invariance test statistic conditioned
on the random vector $X\in S^{n-1}$ is denoted as
$\hat{T}_h : \mathcal{G} \rightarrow \real$ and is defined as
$
 \hat{T}_h(M) = \TT{(MX)}A^h(MX)
$
for $M\in \mathcal{G}$ and
$A$ is the shift matrix as defined 
above. 
Here and in what follows, $\mathcal{G}$ will
be one of $SO(n)$, $\mathbb{S}_n$, or
$\mathbb{B}(n)$.
Of course, it is assumed that $h<n$ so that
the above is nontrivial and, in
particular, we assume that $h\ll n$.
As noted above, we can, without loss of generality, 
assume that $X \in S^{n-1}$. The randomization hypothesis
that $MX = X$ in distribution for any $M\in \mathcal{G}$ 
equates to the null hypothesis that $X$ is isotropic for 
$\mathcal{G}=SO(n)$, a natural setting in statistics; 
i.e. $\xv (X\TT{X}) = \sigma^2 I$.  
For the test statistic
$\hat{T}_h$, we are looking for a specific deviation 
from isotropy, which is non-zero correlation along the
$h$-diagonal of $\xv X\TT{X}$.  

The null hypothesis we aim to test is that of zero
autocorrelation at a given lag $h$.  That is,
\begin{equation}
  \label{hyp:uncor}
  H_0:  K(h) := \cov{X_t}{X_{t+h}} = 0.  
\end{equation}
This implies that the $h$ off diagonal of the 
autocovariance matrix is zero. 
Thus, this hypothesis is more specific than the broader 
total randomization hypothesis that $X \eqdist MX$ 
for any choice of $M\in SO(n)$.
In particular, the total randomization hypothesis
implies that the autocovariance matrix is $cI$ for
some $c>0$.
Therefore, instead of testing for the rotational 
invariance of $X$, we instead test for the 
rotational invariance of the quadratic form
$
 \TT{X}A^hX,
$
which is
\begin{align}
  \label{hyp:rotInv}
  H_0: \xv_X [\TT{X}A^hX] &= \xv_X \xv_{M} [\TT{(MX)}A^h(MX)] \\
  \nonumber
  \xv_X \hat{T}_h(I) &= \xv_X \xv_M \hat{T}_h(M)
\end{align}
where $I$ is the identity matrix. We note that
$\xv_X [\TT{X}A^hX] = (n-h)K(h)$.
Integrating over $M$ using 
Lemma~\ref{lem:rotInnerProduct} in the appendix gives
that 
$\xv_{M} [\TT{(MX)}A^h(MX)] = 0$ almost surely.
Hence, null hypotheses \ref{hyp:uncor} and \ref{hyp:rotInv}
are equivalent.

The main theorem of this section 
provides a sub-Gaussian tail bound 
on our random invariance test statistic $\hat{T}_h$.
It follows directly from Corollary~\ref{meck},
Theorems~\ref{thm:lipGen}, \ref{thm:reflectConc},
and \ref{thm:exchangeConc},
and Lemma~\ref{lip} below, and Lemma~\ref{lem:rotInnerProduct} 
in the appendix.
Note that the constants below may be suboptimal.
However, their choice becomes immaterial once the 
beta correction is applied in the following subsection;
see Theorem~\ref{betacon}.

\begin{theorem}\label{con1}
    Let $M$ be a random element of $\mathcal{G}$
    distributed with respect to normalized Haar measure. 
    Then for each $t>0$,
    $$
        \probB{ \hat{T}_h(M) \geq t \,|\, X} \leq \left\{
        \begin{array}{ll}
          \exp{
            \{-{(n-2)t^2}/{64}\}
          }, & \text{ if } \mathcal{G}=SO(n)\\
          \exp{
            \{-{t^2}/{16}\}
          }, & \text{ if } \mathcal{G}=\mathbb{S}(n)\\
          \exp{
            \{-{t^2}/{64}\}
          }, & \text{ if } \mathcal{G}=\mathbb{B}(n)
    \end{array}
        \right..
    $$
\end{theorem}

We test the hypothesis, $H_0 : \eta_h = 0$ vs $H_1 : \eta_h > 0$; 
i.e. we start by detecting if there is any positive
autocorrelation in our data. 
From our notation, 
$\hat{T}_h(I)$ is based on the original (untransformed) data 
where $I$ is the identity matrix. 
Then the p-value of our randomization hypothesis test is 
$\probB{\hat{T}_h(M) \geq \hat{T}_h(I)\,|\, X}$. 
One way to approximate this is by randomly 
generating a large number of random transformations
from $\mathcal{G}$. 
However, computations of such a large order will be 
impractical especially for $SO(n)$ as noted above in the introduction. 
Therefore, Theorem~$\ref{con1}$ 
allows us to avoid relying on simulation-based 
approximations and provides a sub-Gaussian 
bound for the p-value.
While this bound is too weak for practical 
application as it stands, it will be strengthened
in the next section via a clever transformation.

%Let $M \in \mathbb{SO}(n)$ and $x \in \mathbb{R}^n$. Then we have, \begin{equation}
%\|Mx\|_2^2=(Mx)^tMx=x^tM^tMx=x^tx=\|x\|_2^2
%\end{equation}
%i.e., orthogonal matrices preserve the $L^2$ norm and as a consequence, the operator norm as well. Therefore, we can conclude that $MX$ also belongs to $S^{n-1}$

%As noted, Monte Carlo generation of uniform elements
%from $SO(n)$ will be computationally cumbersome for
%large $n$.  
%Hence, we avoid the need to sample 
%uniform elements from $SO(n)$ by employing the 
For $SO(n)$,
Theorem~\ref{con1} directly stems from the 
concentration inequality for real 
valued Lipschitz functions on $SO(n)$
as stated in Theorem 5.17 in \cite{MECKES2019}
and, more specifically, in Corollary~\ref{meck} above.
The mean $\xv_{SO(n)} [\hat{T}_h(M)] = 0$ for any lag $h$,
which follows directly from 
Proposition~\ref{lem:rotInnerProduct} in 
the appendix.
Thus, Theorem~\ref{con1} is immediately 
established by determining the
Lipschitz constant for our rotation test statistic
$\hat{T}_{h,x}: SO(n) \rightarrow \real$ defined as
$\hat{T}_{h,x}(M) = \TT{(Mx)}A^h(Mx)$.  This comes from the
following lemma with proof.

\begin{lemma}\label{lip}
For a fixed $x \in S^{n-1}$, the function $\hat{T}_{h,x}$ 
is Lipschitz with Lipschitz constant $2$, where the group $\mathrm{SO}(n)$ is equipped with the Hilbert-Schmidt metric.
\end{lemma}

The proof of this lemma follows directly from the proof of
Theorem~\ref{thm:lipGen}.
Specifically, from the Gershgorin circle theorem,
    the eigenvalues of 
    $B^h = (A^h +(A^h\TT{)})/2$ are 
    all bounded above by 1 in magnitude for any choice
    of $h<n$.  More explicit eigenvalue computations
    can be found in \cite{toep} and discussed in Section~\ref{sec:exactChiSquare}
    below.

While the main focus of this work is randomization testing with 
respect to the group $SO(n)$, we also note for completeness
that under 
$\mathbb{S}(n)$ with the assumption that 
$\sum_{t=1}^n X_t=0$ and $\sum_{t=1}^n X_t^2=1$,
$$
  \xv_{\mathbb{S}(n)}[ \hat{T}_{h}(M) ] 
  = \frac{1}{n!}\sum_{\pi\in\mathbb{S}(n)}
    \sum_{t=1}^{n-h} X_{\pi(t)}X_{\pi(t+h)}
  = \frac{1}{n!}\sum_{t=1}^{n-h} \sum_{s\ne r} (n-2)!X_sX_r
  = -\frac{n-h}{n(n-1)}
$$
as 
$0 = [\sum_{s=1}^n X_s]^2 = 
\sum_{s=1}^n X_s^2 + \sum_{s\ne r} X_sX_r$ in this setting.
Hence, in Theorem~\ref{con1} under $\mathbb{S}(n)$,
$$
  \probB{ \hat{T}_h(M) \geq t} \leq
  \probB{ \hat{T}_h(M) + \frac{n-h}{n(n-1)} \geq t} \leq
  \exp{
            \{-{t^2}/{48}\}
          }.
$$
with the mean vanishing at an $n^{-1}$ rate.
In contrast, under $\mathbb{B}(n)$,
$$
  \xv_{\mathbb{B}(n)}[ \hat{T}_{h}(M) ] 
  = 2^{-n}\sum_{b\in \{\pm1\}^n}
    \sum_{t=1}^{n-h} b_tb_{t+h}X_{t}X_{t+h} = 0.
$$

\subsubsection{Power Against Local Alternatives}
\label{sec:powerAnal}

In the context of Theorem~\ref{con1} 
specifically for $SO(n)$ above, we can consider
the asymptotic statistical power for detecting 
deviations from the null under a sequence 
of local alternatives such that $\eta_h\rightarrow0$ as
$n\rightarrow\infty$.  This results in the 
following quick corollary.

\begin{corollary}
    \label{cor:powerAnal}
    Let $X\in\real^n$ be a random vector with autocorrelation
    $\eta_h$ at lag $h$ and finite 
    fourth moment, and let
    $$
      pv_\vee(X) := 
      \exp{\{{-(n-2)(\hat{T}_{h,X}(I))^2}/{64}\}}
    $$
    be a random p-value
    under $\mathcal{G} = SO(n)$.
    If $\eta_h = \delta n^{-1/2}$ for any
    $\delta>0$, 
    then $pv_\vee(X)\convd Z\in(0,1)$ a random variable 
    as $n\rightarrow\infty$.
    Furthermore,
    if $\eta_h = \delta n^{-1/2+\veps}$ for any
    $\veps,\delta>0$, 
    then $pv_\vee(X)\convas 0$ as $n\rightarrow\infty$.
\end{corollary}

\subsubsection{Exact p-values for random rotations}
\label{sec:exactChiSquare}

The classic Durbin-Watson statistic yields a weighted sum of 
chi-squared random variables under the null hypothesis.
The right-hand side of the below Proposition~\ref{prop:exact} is the
distribution function of a linear combination of independent
$\chi^2(1)$ random variables and can be evaluated to any prescribed
accuracy by classical numerical inversion of the characteristic
function \citep{IMHOF1961, davies1980}, implemented in the
\texttt{CompQuadForm} R package \citep{duchesne2010}; for p-values in
the extreme tail, where inversion methods with absolute error bounds
lose accuracy, the saddlepoint approximation of \citet{kuonen1999}
provides relative accuracy.

However, this requires specific knowledge of the spectrum of the
quadratic form.  
In contrast, our proposed beta-corrected statistic detailed in the
next section
is a closed-form
expression requiring no numerical inversion. Compared to
Proposition~\ref{prop:exact}, our approach extends 
unchanged to the groups
$\mathbb{S}(n)$ and $\mathbb{B}(n)$ and to arbitrary 
symmetric matrices $A$ whose spectra
are not available in closed form.

\begin{proposition}\label{prop:exact}
Fix $x \in S^{n-1}$, let $M$ be Haar distributed on $SO(n)$, let
$z \sim \mathcal N(0, I_n)$, and let $\lambda_1, \dots, \lambda_n$
denote the eigenvalues of $B_h = (A^h + (A^h)^{\mathrm T})/2$. Then
$Mx \overset{d}{=} z / \lVert z \rVert$ and, for any $t \in \mathbb R$,
$$
  \prob{ \hat T_{h,x}(M) \ge t } = 
  \prob{ \sum_{i=1}^{n} (\lambda_i - t)\, W_i \ \ge\ 0 },
  \quad W_1, \dots, W_n \distiid \chi^2(1).
$$
Furthermore, writing $n = qh + r$ with $0 \le r < h$, the eigenvalues
of $B_h$ are
\begin{align*}
  \cos\left( \frac{j\pi}{q+2} \right),~& j = 1, \dots, q+1,
  \text{ each with multiplicity } r,
  \text{ and}\\
  \cos\left( \frac{j\pi}{q+1} \right),~& j = 1, \dots, q,
  \text{ each with multiplicity } h - r.    
\end{align*}
\end{proposition}

\subsection{Beta Correction}\label{beta}
\label{sec:betaCorrect}

A significant challenge in employing concentration inequalities
for direct statistical application 
is the substantial loss of power to reject the null hypothesis 
caused by universal constants that are often too big 
for applications. 
Motivated by the approach in \cite{KASHLAK_KHINTCHINE2020,KASHLAK_YUAN_ABELECT}, 
we apply a transformation based on the 
beta distribution to rectify our p-values and 
restore the lost statistical power.  
The main theorem
of this section intuitively stems from 
the probability integral transform; that is, a non-atomic 
real random variable $Z$ with cumulative distribution
function $F_Z$ results in 
$F_Z(Z)\dist\distUnifInt{0}{1}$.  The claim of the 
following theorem is that a random variable plugged 
into its sub-Gaussian tail bound will elicit a distribution
close to a beta distribution.

\begin{remark}
The following theorem is stated for $\mathcal{G}=SO(n)$, but effectively
all that matters is the form of the right hand side of the
concentration inequality.  Similar formulations for 
$\mathbb{S}(n)$ and $\mathbb{B}(n)$ will give similar p-values.
\end{remark}

\begin{theorem}\label{betacon}
    In the setting of Theorem $\ref{con1}$ with
    $\mathcal{G}=SO(n)$ and additionally that 
    $X_t$ is a stationary time series with 
    iid mean zero noise process $\veps_t$ with
    finite 3rd moment,
    let 
    $$
      pv_\vee = 
      \exp{\{{-(n-2)(\hat{T}_{h,X}(I))^2}/{64}\}}
    $$ be the upper bound on the tail probability (p-value)
    in terms of random vector $X$ from Theorem~\ref{con1}.
    Then, %for $n$ sufficiently large,
    \begin{equation}
      \label{eqn:betaPV}
      \mathbb{P}\left(
        pv_\vee \leq u \right) \leq 
        C_0%\left\{
        %\begin{array}{ll}
          I\left[u;\frac{32n(n+2)}{(n-h)(n-2)},\frac{1}{2}
          \right]  % & \text{ if }
         % \mathcal{G} = SO(n)\\
         % I\left[u;\frac{24n(n+2)}{n-h},\frac{1}{2}
         % \right]   & \text{ if }
         % \mathcal{G} = \mathbb{S}(n)\\
         % I\left[u;\frac{32n(n+2)}{n-h},\frac{1}{2}
         % \right] & \text{ if }
         % \mathcal{G} = \mathbb{B}(n)
      %\end{array}  
      + O(n^{-1/2})
      %\right.
    \end{equation}
    where $I(u; \alpha, \beta)$ is the regularized incomplete beta function and,
    $$
      C_0 = %\left\{
      %\begin{array}{ll}
        \left(\frac{32n(n+2)}{(n-h)(n-2)}\right)^{1/2}
        \Gamma \left(\frac{32n(n+2)}{(n-h)(n-2)}\right) 
        \Gamma \left(\frac{1}{2} + 
        \frac{32n(n+2)}{(n-h)(n-2)}\right)^{-1}
       % & \text{ if }\mathcal{G} = SO(n) \\
       % \left(\frac{24n(n+2)}{n-h}\right)^{1/2}
       % \Gamma \left(\frac{24n(n+2)}{n-h}\right) 
       % \Gamma \left(\frac{1}{2} + 
       % \frac{24n(n+2)}{n-h}\right)^{-1}
       % & \text{ if }\mathcal{G} = \mathbb{S}(n) \\
       % \left(\frac{32n(n+2)}{n-h}\right)^{1/2}
       % \Gamma \left(\frac{32n(n+2)}{n-h}\right) 
       % \Gamma \left(\frac{1}{2} + 
       % \frac{32n(n+2)}{n-h}\right)^{-1}
       % & \text{ if }\mathcal{G} = \mathbb{B}(n) 
      %\end{array}
      %\right.
    $$
    For large $n$, $C_0\approx 1.004$.
    %for $SO(n)$ and 
    %$C_0\rightarrow1$ for $\mathbb{S}(n)$ and $\mathbb{B}(n)$.
\end{theorem}

To establish this theorem, we begin 
by calculating the mean and variance of the 
autocorrelation estimator, $\hat{T}_{h,X}(I)$, 
under $H_0$. Under $H_0$,
$
    \mathrm{E}(X_iX_{i+h}) = 0
$
and hence $\xv \hat{T}_{h,X}(I) = 0$.
The above is for the identity 
matrix $I$ and a uniformly random unit vector $X$.
However, this is equal in distribution to fixing
a unit vector $x$ and considering $\hat{T}_{h,x}(M)$
for $M$ a Haar-distributed random rotation matrix from $SO(n)$.
In order to calculate the variance, 
we will need the following lemma from \cite{wiens}, 
which is stated without proof.
\begin{lemma}\label{var}
    Let $X$ be an exchangeable random vector. 
    %with a spherical distribution. 
    Then for symmetric matrices $A$ and $B$, we have
    \begin{equation}
        \mathrm{E}(\TT{X}AX  \TT{X}BX) = 
        \mu_{22}[\mathrm{tr}(A)\mathrm{tr}(B) + 2\mathrm{tr}(AB)]
        +
        (\mu_4 - 3\mu_{22})\tr{ A\circ B }
    \end{equation}
    where $\mu_{22} = \mathrm{E}(X_1^2X_2^2)$,
    $\mu_4 = \xv{X_1^4}$, and $\circ$ is the Hadamard
    product.
\end{lemma}

A spherical distribution means that ${x}/{||x||}$ 
is distributed uniformly over the surface of the unit 
sphere independently of $||x||$. This result is crucial 
for our work, as it allows us to easily calculate 
expectations of products of quadratic forms, 
aiding us in the computation of the variance
in the following lemma with proof in 
Appendix~\ref{app:proofs2}.

\begin{lemma}
   \label{lem:nullVar}
   Under $H_0$, the variance of 
   $\hat{T}_{h,X}(I)$ for $X$ uniform on $S^{n-1}$ 
   and equivalently 
   $\hat{T}_{h,x}(M)$ for a fixed $x\in S^{n-1}$
   and $M$ uniform on $SO(n)$
   is ${(n-h)}/{n(n+2)}$.
\end{lemma}

\noindent
Therefore for large $n$ and fixed $h\ll n$, the variance 
is of the order ${1}/{n}$.  
Finally, combining the above results
gives the proof of our main theorem, 
which is detailed in Appendix~\ref{app:proofs2}.

Theorem \ref{betacon} allows us to improve our p-values from Theorem \ref{con1} in  order to recover the lost statistical power. The improved bound is,
$$
    \probB{\hat{T}_{h,X}(M) \geq t}
    \leq 
    C_0I\left[\exp{
      \bigg\{\frac{-(n-2)t^2}{64}
    \bigg\}};\frac{32n(n+2)}{(n-1)(n-2)},\frac{1}{2}
    \right].
$$
Since $t > 0$, this test essentially tests for positive autocorrelation instead of any autocorrelation in the data. Hence, this is a one-sided test. In order to make it a two-sided test, we need a bound for $\probB{|\hat{T}_{h,X}(M)| \geq t}$. For $t>0$,
\begin{align*}
    \probB{|\hat{T}_{h,X}(M)| \geq t} 
    &= \probB{
      \hat{T}_{h,X}(M) \geq t} + \probB{\hat{T}_{h,X} \leq -t}\\
    &= \probB{
      \hat{T}_{h,X}(M) \geq t} + \probB{-\hat{T}_{h,X}(M) \geq t}.
\end{align*}
By similar argument, $-\hat{T}_{h,X}(M)$ 
is also 2-Lipschitz, and we have the same 
concentration inequality for $-\hat{T}_{h,X}(M)$ as well. 
Therefore, we can conclude that
\begin{equation}\label{final}
    \probB{|\hat{T}_{h,X}(M)| \geq t} \leq 2C_0I\left[\exp{\bigg\{\frac{-(n-2)t^2}{64}\bigg\}};\frac{32n(n+2)}{(n-h)(n-2)},\frac{1}{2}\right],
\end{equation}
which
results in a two-sided test.

\subsection{Uncorrelated Testing of Multiple Lags}
\label{sec:multiTest}

The classic Durbin-Watson test does not extend beyond
lag 1.  Meanwhile, the Breusch-Godfrey test tests for 
autocorrelation at all lags from $1$ to $h$.  In contrast,
our proposed test based on the value $\hat{T}_{h,X}(I)$
tests only for autocorrelation specifically
at a lag $h$.  Furthermore,
the statistics $\hat{T}_{h_1,X}(I)$ and $\hat{T}_{h_2,X}(I)$
are uncorrelated for $h_1\ne h_2$.  
\begin{proposition}
\label{prop:uncorr}
  For $
    \hat{T}_{h,X}(M) = \TT{(MX)}A^h(MX),
  $
  let $T_1 = \hat{T}_{h_1,X}(I)$ and $T_2 = \hat{T}_{h_2,X}(I)$.
  Then, $\cov{T_1}{T_2}=0$.
\end{proposition}

\noindent
Hence, we can recover a 
test for autocorrelation at an arbitrary  
collection of lags $h\in H \subset \{1,2,\ldots,m\}$
with $m \ll n$.

If the original data is normally distributed or 
the sample size $n$ is sufficiently large, then 
Fisher's method of combining independent p-values
can be used.  
Moreover, 
under the global white
noise null setting with finite fourth moments, the vector
$(\sqrt n\, \hat\eta_h)_{h \in H}$ converges jointly to a standard
multivariate normal distribution \citep{brockwell2009},
so the per-lag p-values are asymptotically mutually independent,
Otherwise, multiple testing correction
methods can be employed for testing for autocorrelation
at multiple lags.  These are considered in the 
simulated data experiments in 
Section~\ref{sec:sims}.

\section{Simulated Data Analysis}
\label{sec:sims}

\subsection{Simulation Experiment Overview}

In the following simulation studies, we 
compare the power of our random rotation 
test both before (Section~\ref{sec:subGauss})
and after the proposed beta correction 
(Section~\ref{sec:betaCorrect})
against the classic Durbin-Watson test
for lag 1 autocorrelation
and the omnibus Breusch-Godfrey and 
Ljung-Box tests for autocorrelations
at different lags.
Of note, our test tests for autocorrelation at
a given lag $h$ whereas Breusch-Godfrey and 
Ljung-Box test for
autocorrelation at all lags $1,\ldots,h$.  
The Durbin-Watson test is performed via 
\texttt{dwtest} and the Breusch-Godfrey 
test via \texttt{bgtest} both in the 
\texttt{lmtest} R package \citep{LMTEST}.
The Ljung-Box test was performed via the 
\texttt{Box.test} function in the base
\texttt{stats} R package.

To summarize the following simulation experiments
for autoregressive processes, 
the Durbin-Watson test 
demonstrates superior power at lag 1 but is still 
close to our beta-corrected rotation test.
Beyond lag 1,
our single lag beta-corrected rotation test achieves 
stronger statistical power than the 
multi-lag Breusch-Godfrey and
Ljung-Box tests.  The uncorrected sub-Gaussian
bound from Section~\ref{sec:subGauss} 
is the weakest of the tests considered, which is
to be expected as the concentration inequalities
implemented cover very general settings.  This 
also emphasizes the utility of our beta-correction
technique.
When multiple significant lags exist in the data,
our rotation p-values combined with Fisher's method
as detailed in Section~\ref{sec:multiTest} achieves
the highest power while still maintaining the correct
empirical test size.  
For moving average processes, our single lag rotation
tests have worse power than the classic tests.  However,
combining all of our p-values using Fisher's method 
results in a most powerful test coinciding with 
Breusch-Godfey and Ljung-Box in our empirical
study.
When compared to brute force randomization tests 
via Monte Carlo simulation from one of the three
symmetry groups considered, our concentration-based
rotation test only loses a little power at the 
benefit of zero computational burden. 

In the following simulation experiments, we consider
the behaviour of these statistical tests under heavy
tailed noise, which breaks the finite 4th moment
assumption.  However, this still provides an illustrative 
empirical example of the convergence behaviour of
the sample covariance and correlation 
characterized in \cite{DAVISRESNICK1986}.

\subsection{Null hypothesis with heavy tails}

\begin{table}
    \centering
    \begin{tabular}{r|rrrr}
    \hline\hline
         & \multicolumn{2}{c}{$\distTdist{4}$} &
           \multicolumn{2}{c}{$\distNormal{0}{1}$} \\
     lag & KS test & AD test & KS test & AD test\\ \hline
       1 &   0.540 &   0.554 & 0.343   & 0.069 \\
       2 &   0.086 &   0.207 & 0.181   & 0.109 \\
       3 &   0.086 &   0.126 & 0.690   & 0.559 \\
       4 &   0.751 &   0.639 & 0.440   & 0.573 \\
       5 &   0.023 &   0.066 & 0.304   & 0.100 \\
       \hline\hline
    \end{tabular}
    \caption{
      The p-values for the Kolmogorov-Smirnov (KS) test
      and the Anderson-Darling (AD) test for goodness-of-fit
      for our 5000 rotation test p-values (Theorem \ref{betacon})
      compared with the $\distUnifInt{0}{1}$ distribution.
    }
    \label{tab:nullSize}
\end{table}

\begin{figure}
    \centering
    \includegraphics[width=0.45\linewidth]{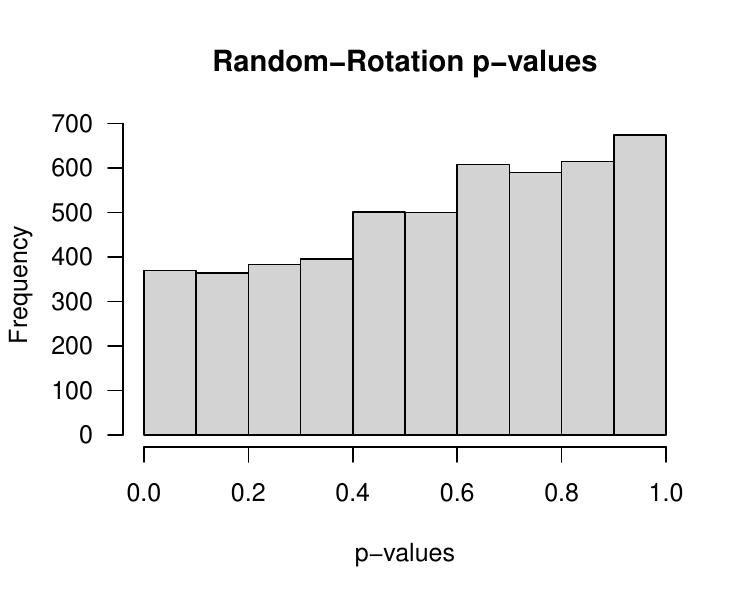}
    \includegraphics[width=0.45\linewidth]{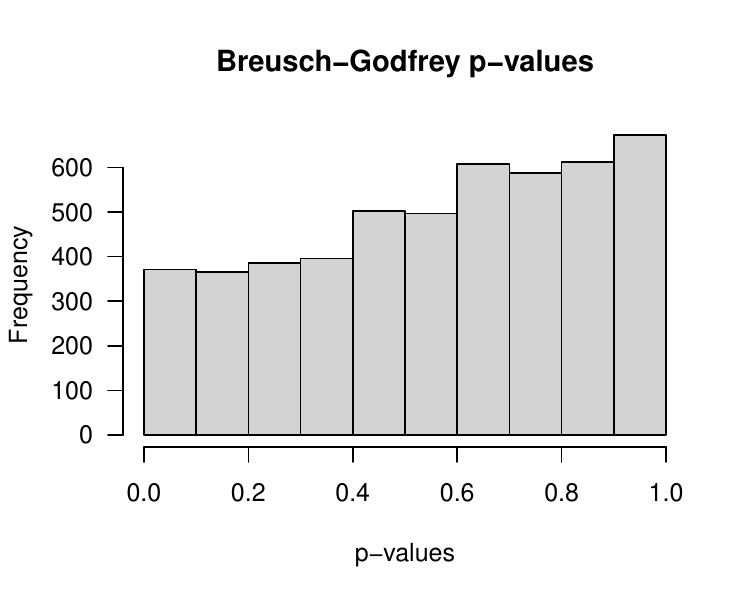}
    \caption{
      Histograms of 5000 null p-values produced by 
      the random-rotation test (left) and the Breusch-Godfrey
      test (right) for the $\distTdist{2}$ distribution.
      Both are slightly conservative in the presence of 
      very heavy tailed noise.
    }
    \label{fig:nullSize}
\end{figure}

Before performing a power analysis of our method
compared to the classic tests for autocorrelation,
we test the null distribution of the p-values
generated by our random rotation test after
Beta correction; i.e. the p-values computed
in Theorem~\ref{betacon} by Equation~\ref{eqn:betaPV}.
That is, we generate 5000 replications of 
time series data with $n=1000$ from the 
null model
$
  X_t = \veps_t
$
for $\veps_t$ iid with $\distTdist{4}$ and the
normal distribution. 
Table~\ref{tab:nullSize} details the results of running 
both the Kolmogorov-Smirnov and the 
Anderson-Darling tests for goodness-of-fit 
on the
5000 p-values generated in our simulations 
against the $\distUnifInt{0}{1}$ distribution.
This indicates that our random-rotation test
achieves the correct empirical size under the 
null hypothesis both for Gaussian noise and
for moderately heavy tailed noise.

When $\veps_t$ has such heavy tails that the 
variance is infinite, i.e. $\distTdist{2}$, 
we get conservative p-values from our random-rotation
test.  However, these conservative p-values 
closely coincide with those returned by the 
Breusch-Godfrey test, which also computes conservative
p-values in the presence of very heavy tailed noise.
Histograms of these p-values for $\distTdist{2}$
noise are displayed in Figure~\ref{fig:nullSize}.

\subsection{AR(1) processes with heavy tails}

\begin{figure}
    \centering
    \includegraphics[width=0.48\linewidth]{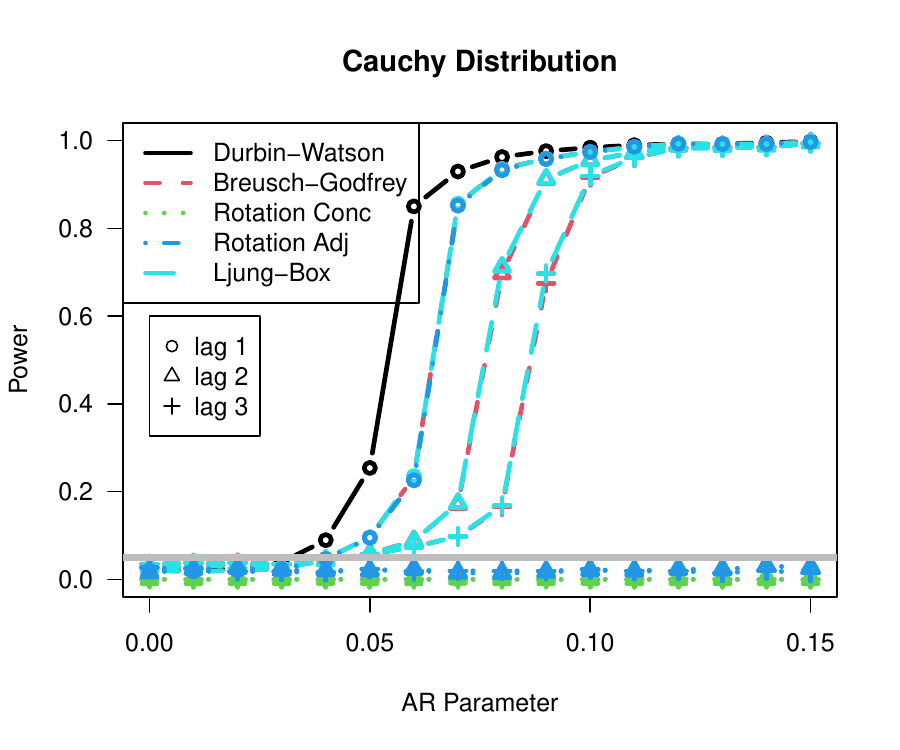}
    \includegraphics[width=0.48\linewidth]{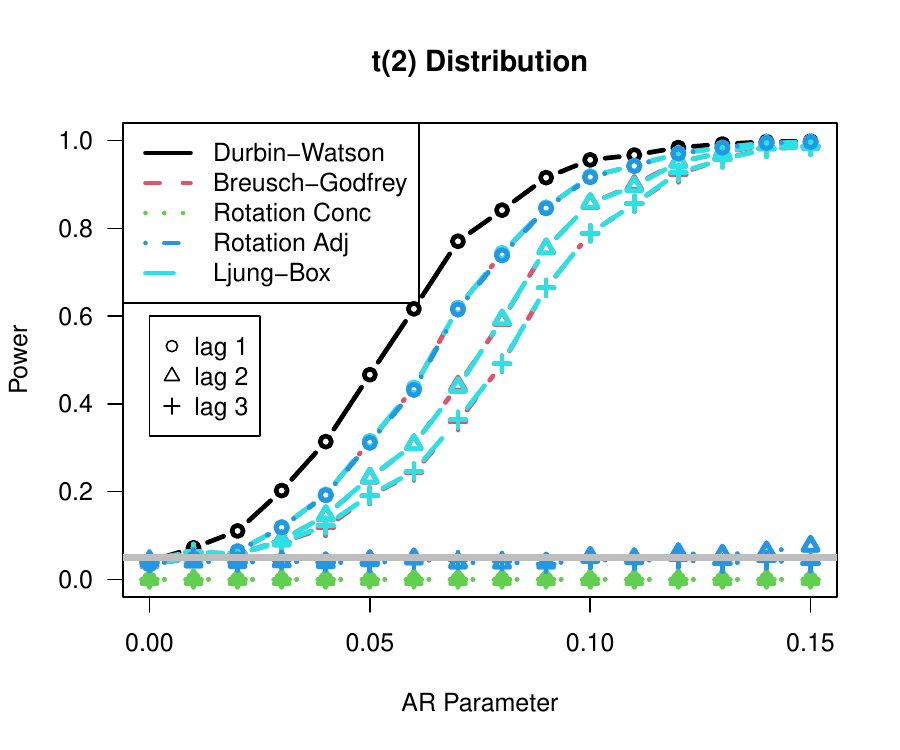}
    \includegraphics[width=0.48\linewidth]{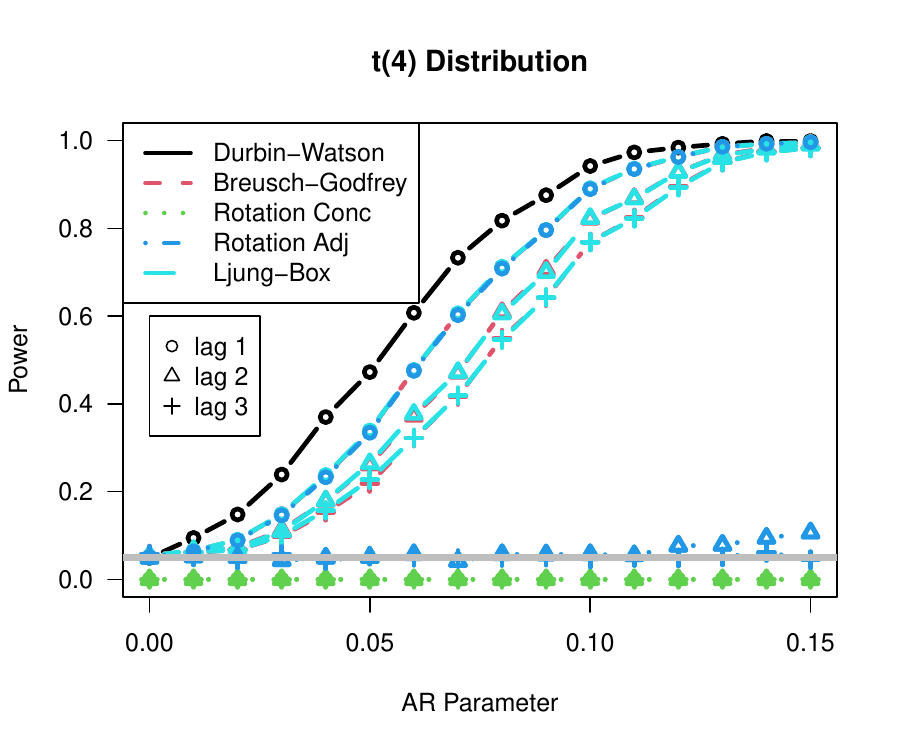}
    \includegraphics[width=0.48\linewidth]{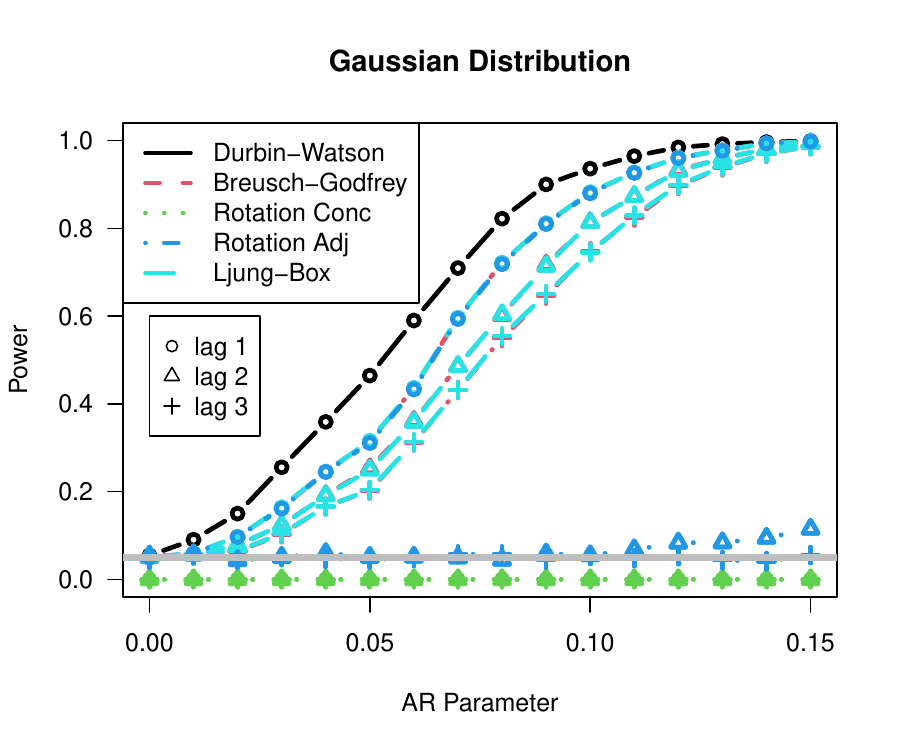}
    \caption{A comparison of the empirical power of 
    Durbin-Watson, Breusch-Godfrey, and our rotation-based test
    on the AR(1) process with coefficient $\phi\in[0,0.15]$
    and empirical size $\alpha=0.05$}
    \label{fig:ar1Power}
\end{figure}

For the first set of power experiments, we simulate
2000 time series data sets of length $n=1000$
from an AR(1) model
$
  X_t = \phi X_{t-1} + \veps_t
$
for each $\phi = 0,0.01,\ldots,0.15$.
For the noise process $\veps_t$, we consider
four different distributions:
the t-distribution with degrees of freedom
$1,2,4$ and the normal distribution.
The empirical power curves for these experiments
are displayed in Figure~\ref{fig:ar1Power}.

In all distributional settings including
the Cauchy distribution as an extreme case, 
the Durbin-Watson test achieves the strongest 
statistical power to identify the presence
of autocorrelation at lag 1.  Our rotation
test at lag 1 achieves identical power to 
Breusch-Godfrey and Ljung-Box whereas at higher lags, our
method correctly does not detect any significant
deviation from the null.
All of the Breusch-Godfrey and Ljung-Box tests eventually
return significant p-values but these tests
yield less statistical power than our method 
and Durbin-Watson as they test for 
significance at more lags simultaneously.
Lastly, without applying our Beta correction
from Section~\ref{sec:betaCorrect}, our concentration
of measure approach detailed in Theorem~\ref{con1}
does not yield any statistical power.

\subsection{AR(h) processes with heavy tails}

\begin{figure}
    \centering
    \includegraphics[width=0.48\linewidth]{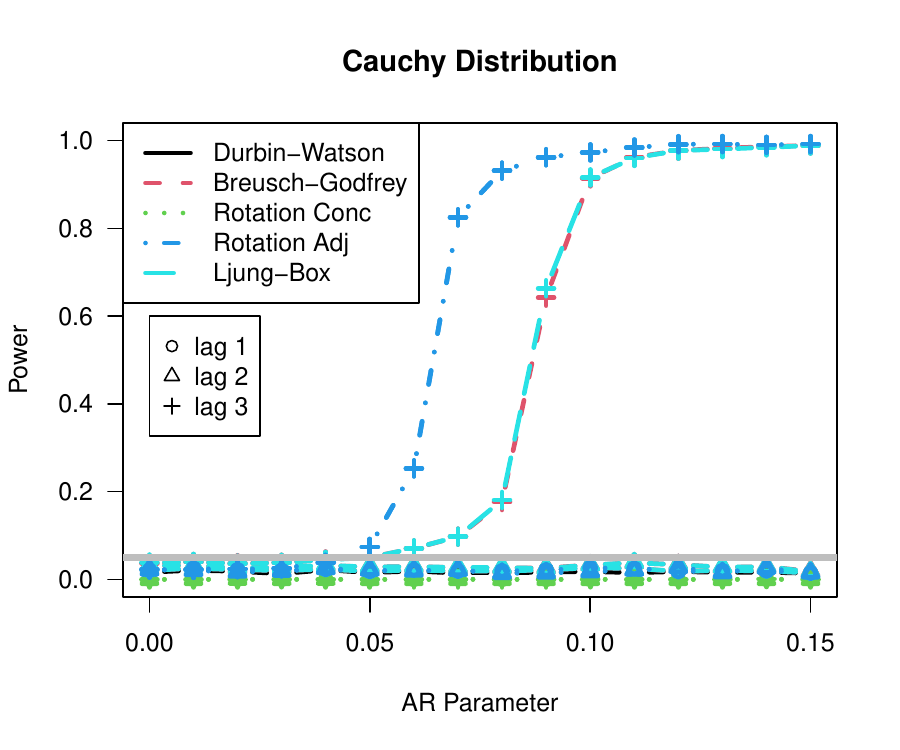}
    \includegraphics[width=0.48\linewidth]{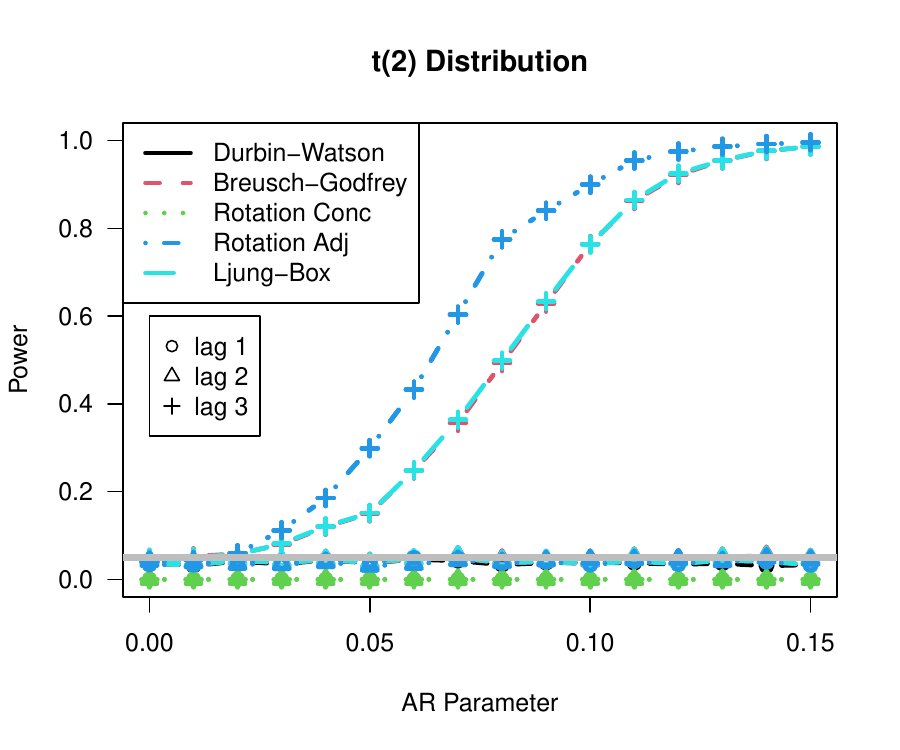}
    \includegraphics[width=0.48\linewidth]{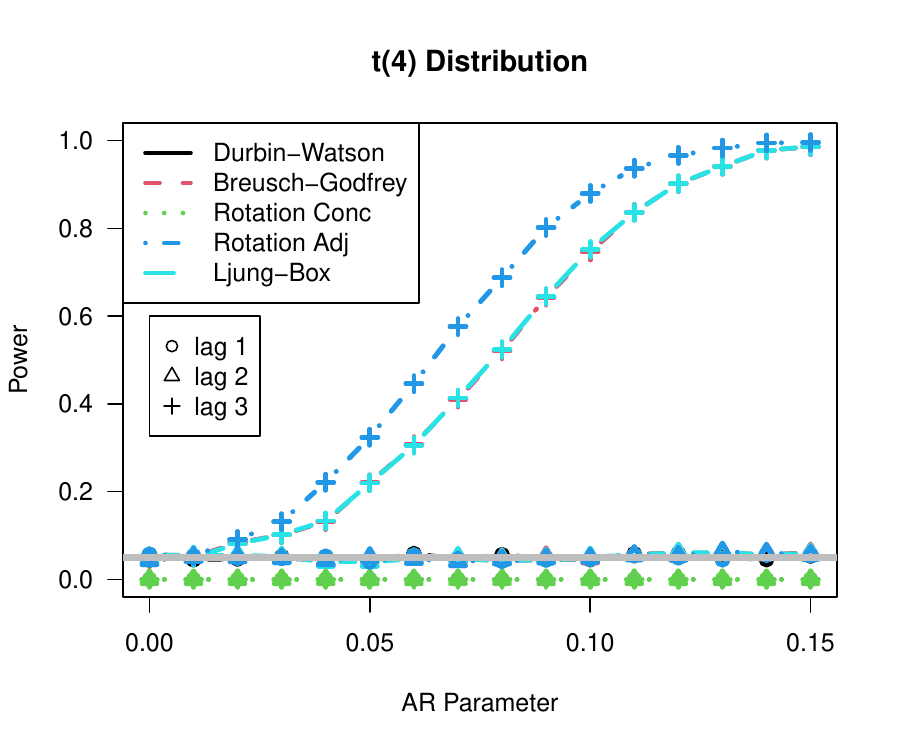}
    \includegraphics[width=0.48\linewidth]{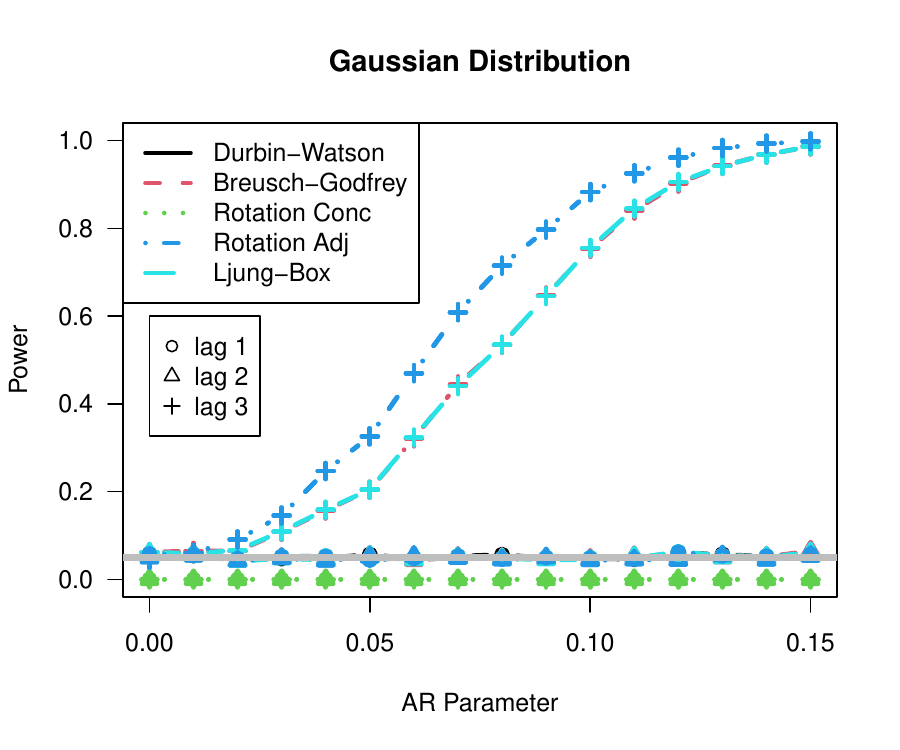}
    \caption{A comparison of the empirical power of 
    Durbin-Watson, Breusch-Godfrey, and our rotation-based test
    on the AR(3) process with coefficient $\phi\in[0,0.15]$
    and empirical size $\alpha=0.05$. }
    \label{fig:arhPower}
\end{figure}

For our second set of experiments for statistical 
power, the above
setting is repeated except the lag for the 
autocorrelation is increased from 1 to $h$
resulting in
an AR(h) model
$
  X_t = \phi X_{t-h} + \veps_t.
$
In this setting, the Durbin-Watson test 
(correctly) does not detect any significant
autocorrelation at lag 1.  Our method
run at lag $h$ and the Breusch-Godfrey and Ljung-Box tests
run at lags $\ge h$ all identify the presence
of autocorrelation.  However, our rotation
test post-Beta correction 
achieves the strongest statistical power 
while still maintaining the correct empirical 
test size.
The results of these experiments are all displayed in 
Figure~\ref{fig:arhPower} for $h=3$.

As noted above, both our method and the Breusch-Godfrey
test are a bit conservative when $\veps_t$
has a heavy tailed distribution.  They both 
achieve a similar empirical size, but our 
rotation test achieves higher statistical power
as $\phi$ increases in value.

\subsection{Multi-lag testing}

\begin{figure}
    \centering
    \includegraphics[width=0.48\linewidth]{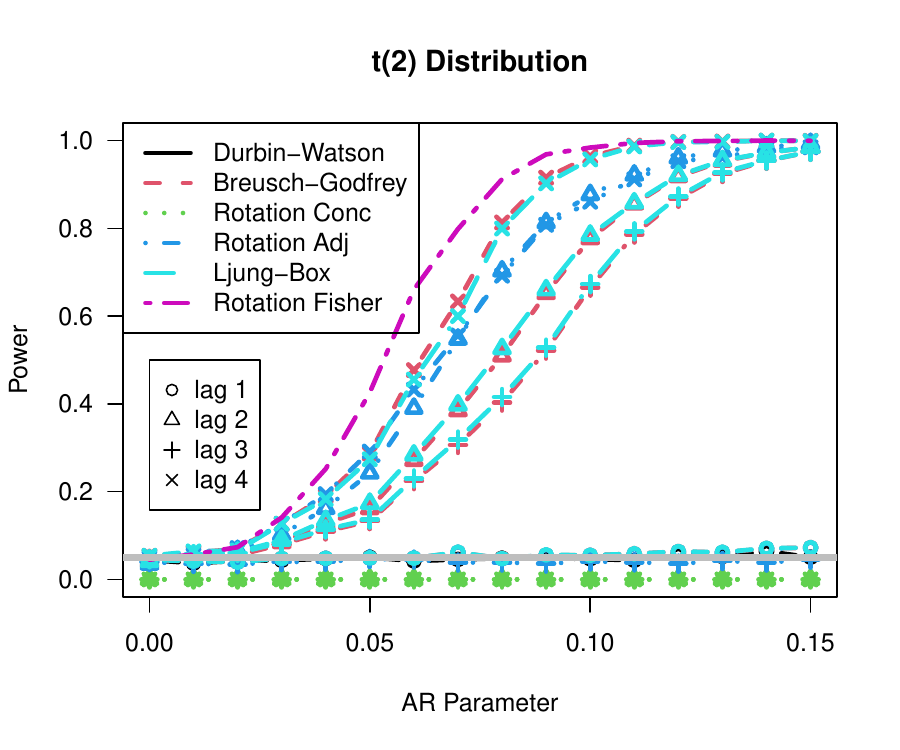}
    \includegraphics[width=0.48\linewidth]{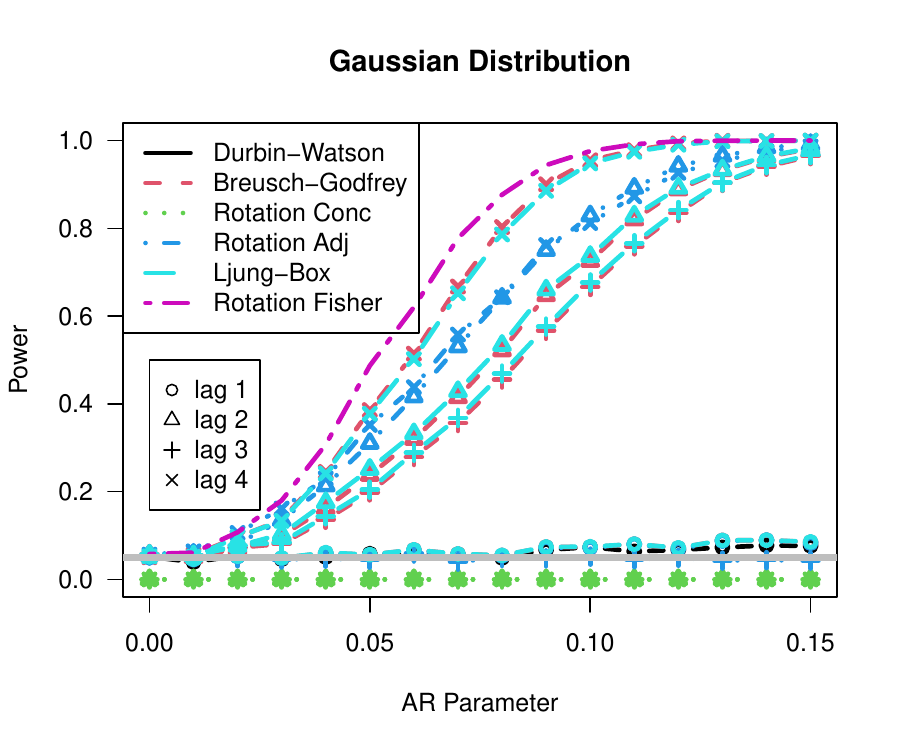}
    \caption{
      A comparison of the empirical power for multiple
      significant lags at $h=2,4$
      and empirical size $\alpha=0.05$.  These plots include
      Fisher's method for combining p-values.
    }
    \label{fig:multLagPower}
\end{figure}

For our third power analysis simulation, we consider 
the following autoregressive process,
$
  X_t = \phi X_{t-2} - \phi X_{t-4} + \veps_t,
$
with $\phi\in[0,0.15]$ and with both normal and 
$\distTdist{2}$ noise processes. 
Once again, $n=1000$, and 2000 such time series are
simulated.
In this scenario, 
our random rotation tests for the null hypotheses
$H_0: \eta_2=0$ and $H_0:\eta_4=0$ yield more statistical
power than 
the Breusch-Godfrey and Ljung-Box tests for null hypotheses
$H_0: \eta_1=\eta_2=0$ and $H_0:\eta_1=\eta_2=\eta_3=0$.
However, our random-rotation test is less powerful than
the Breusch-Godfrey and Ljung-Box tests for null hypotheses
$H_0: \eta_h=0$ for $h=1,\ldots,4$ and 
$H_0: \eta_h=0$ for $h=1,\ldots,5$.

Unlike the Breusch-Godfrey and Ljung-Box tests, which test for significance
at all lags $\le h$, our random-rotation p-values can be 
selectively combined via Fisher's method to test more general
multi-lag hypotheses.  That is, the value
$
  -2 \sum_{h \in H} \log( \text{pv}_h )
$
will test the null hypothesis $H_0: \eta_h = 0$ for $h\in H$ 
with null distribution $\distChiSquared{2\abs{H}}$.
In the setting of these simulations, we test the null
hypothesis
$H_0: \eta_2=\eta_4=0$, which yields superior power to all 
other tests considered on this simulated data.
All of the power curves for these methods are displayed
in Figure~\ref{fig:multLagPower}.

\subsection{Moving Average Processes}

\begin{figure}
    \centering
    \includegraphics[width=0.48\linewidth]{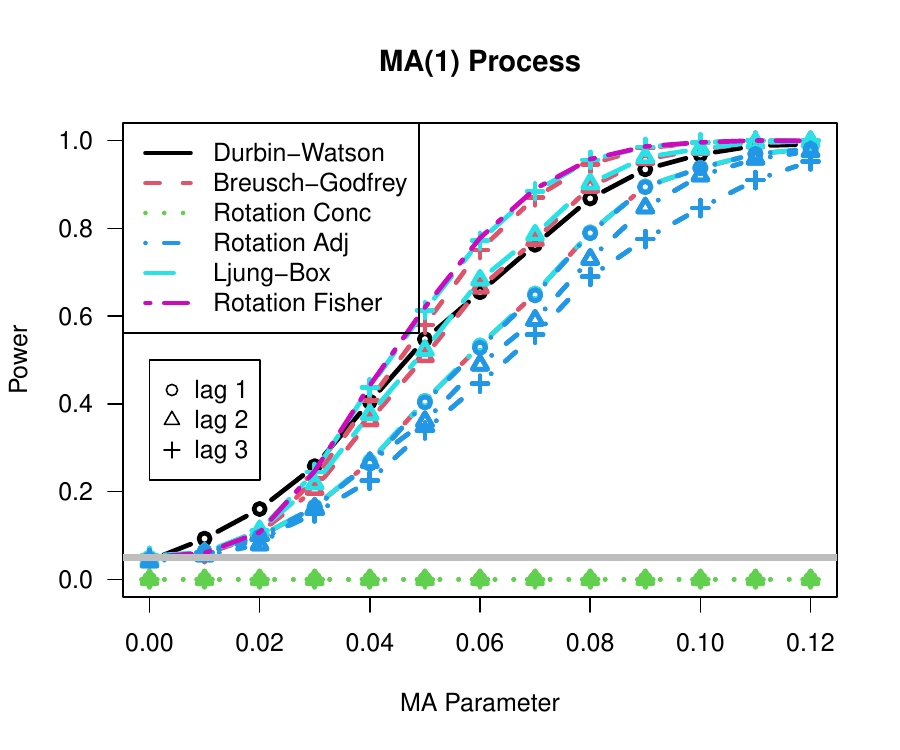}
    \includegraphics[width=0.48\linewidth]{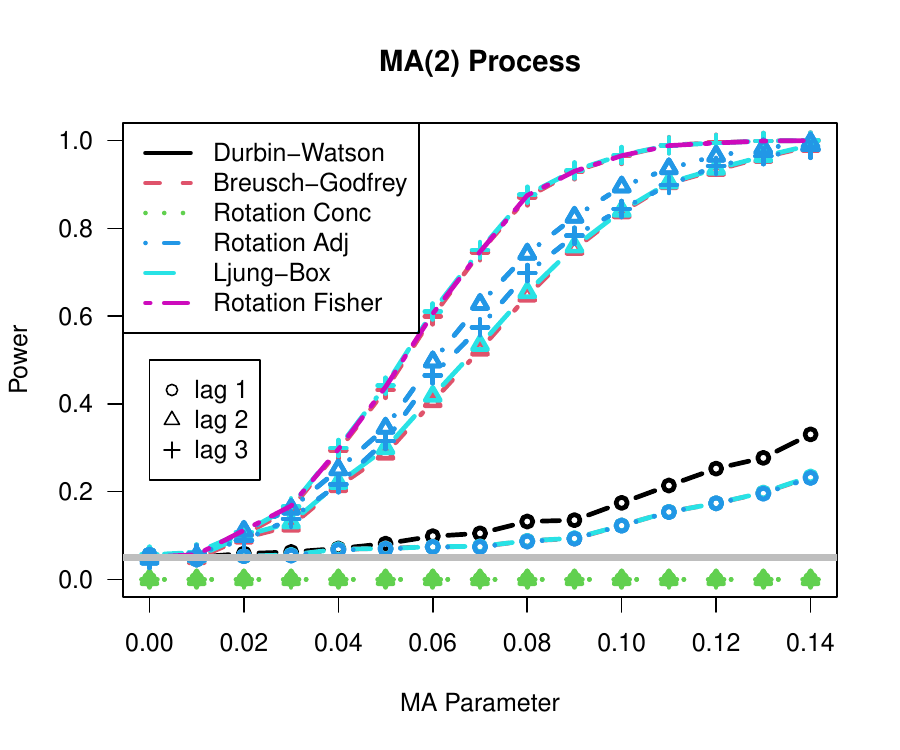}
    \caption{
      A comparison of the empirical power for 
      the moving average process under normal errors
      and empirical size $\alpha=0.05$.  
    }
    \label{fig:maPower}
\end{figure}

A finite moving average process is effectively an infinite
autoregressive process.
We consider the MA(1) process 
$X_t = \veps_t + \theta \veps_{t-1}$
and the MA(2) process
$X_t = \veps_t + \theta \veps_{t-2}$.
Figure~\ref{fig:maPower}
displays a power analysis for these two 
models.
For the MA(1) process, our single lag rotation
tests have lower power than the omnibus 
Breusch-Godfrey and Ljung-Box tests performed at lag 3.
However, we achieve the same statistical power once
we combine our three single lag p-values via 
Fisher's method.
In the case of the MA(2) process, 
similar behaviour is seen.

\subsection{Concentration vs Computation}

\begin{figure}
    \centering
    \includegraphics[width=0.48\linewidth]{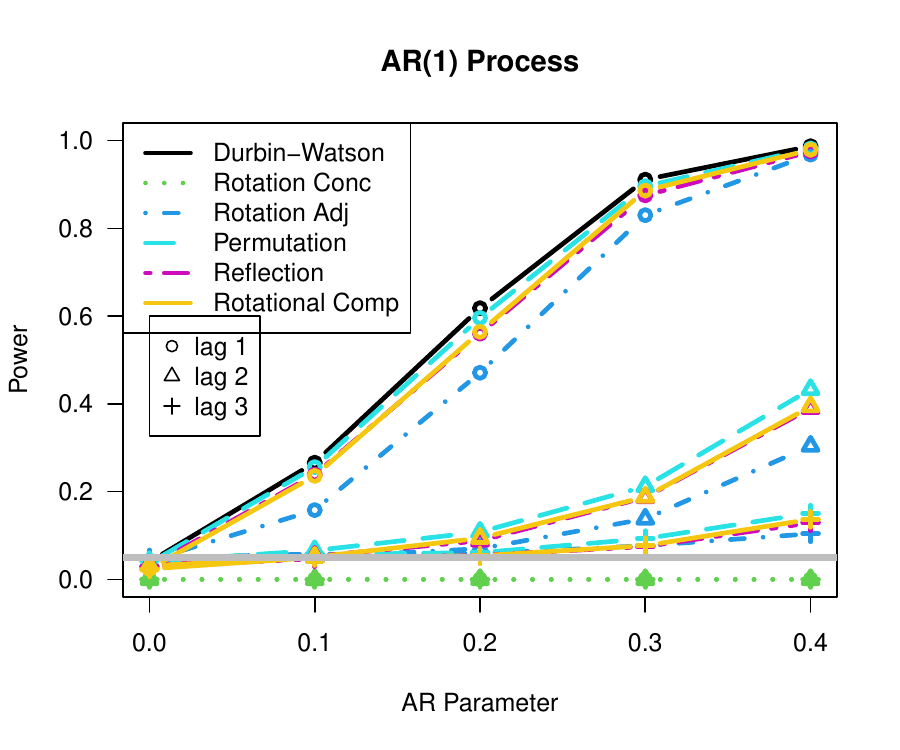}
    \includegraphics[width=0.48\linewidth]{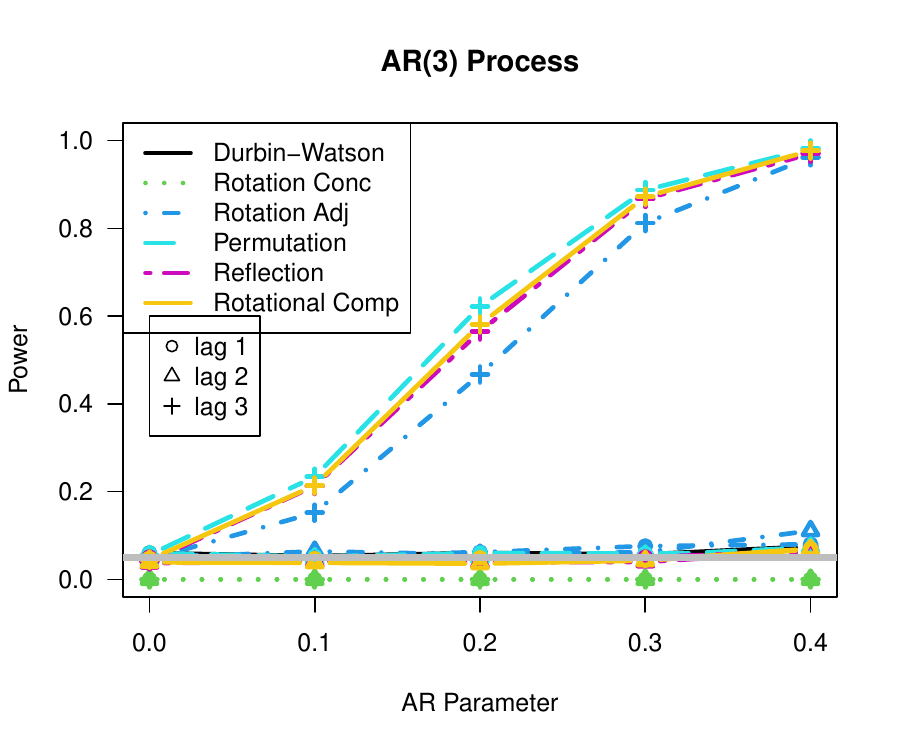}
    \caption{
      A comparison of the empirical power for 
      computation-based randomization tests 
      on the autoregressive 
      process under normal errors
      and empirical size $\alpha=0.05$.  
    }
    \label{fig:compPower}
\end{figure}

In this final simulation, we consider how the adjusted concentration
bound on the p-value from Theorem~\ref{betacon} 
performs against pure computation based randomization tests
that perform Monte Carlo draws from 
$SO(n)$, $\mathbb{S}(n)$, or $\mathbb{B}(n)$.
As these computational tests are significantly slower to perform
than Durbin-Watson, Breusch-Godfrey, Ljung-Box, and our concentration bound, 
in this simulation, we consider samples of size $n=100$ over 
$3000$ repetitions where each randomization test p-value is
computed via simulating $100$ random transformations from each
group of interest.  The two models considered in this 
simulation experiment are the same AR(1) and AR(3) models
from the previous section.  Though, due to the smaller 
sample size considered, we have to increase the value of
the autoregressive parameter much more to detect 
significant autocorrelation with all methods considered.

Figure~\ref{fig:compPower} shows the power analysis of
our analytic method compared to brute force simulation
of random permutations, reflections, and rotation.
We do not include the Breusch-Godfrey or Ljung-Box 
tests here as they
are identical to the previous simulations in the above
sections, and this exclusion makes the plots in 
Figure~\ref{fig:compPower} more intelligible. 
For the AR(1) process, we once again see 
Durbin-Watson achieving the highest statistical power, 
which approximately ties with the computation-based
permutation test.  The computation-based reflection 
and rotation tests have slightly less statistical power.
And finally our concentration-based test performs well
but with another slight drop in power.
We see a similar pattern emerging for the AR(3) model.
The three computation-based randomization tests 
achieve similar statistical power with the permutation
test slightly above the other two tests.
Our concentration-based test only loses a little 
statistical power and the benefit of being 
computationally trivial to compute.

\section{Solar Cycle Data Analysis}
\label{sec:solarData}

To demonstrate the utility of our random rotation 
methodology, we consider the time series of monthly
solar intensity between January 1900 and February 2025 
as provided by the 
National Oceanic and Atmospheric Administration's
Space Weather Prediction Center
(\url{https://www.swpc.noaa.gov/products/solar-cycle-progression}).
The $\log_2$ number of sunspots is plotted in
Figure~\ref{fig:solarIntensity} along with the 
Nadaraya-Watson kernel regression, which was 
fit in R via \texttt{ksmooth} from the base 
\texttt{stats} package with a box kernel and 
bandwidths of 6, 12, and 18 months.  Note that selection
of the smoothing parameter drastically affects which 
lags are deemed statistically significant.
The total size of this time series is $n=1502$ months,
which would make computation-based randomization 
testing very slow to perform.  
The residuals of this
kernel regression model are also plotted in 
Figure~\ref{fig:solarIntensity}.

\begin{figure}
    \centering
    \includegraphics[width=0.85\linewidth]{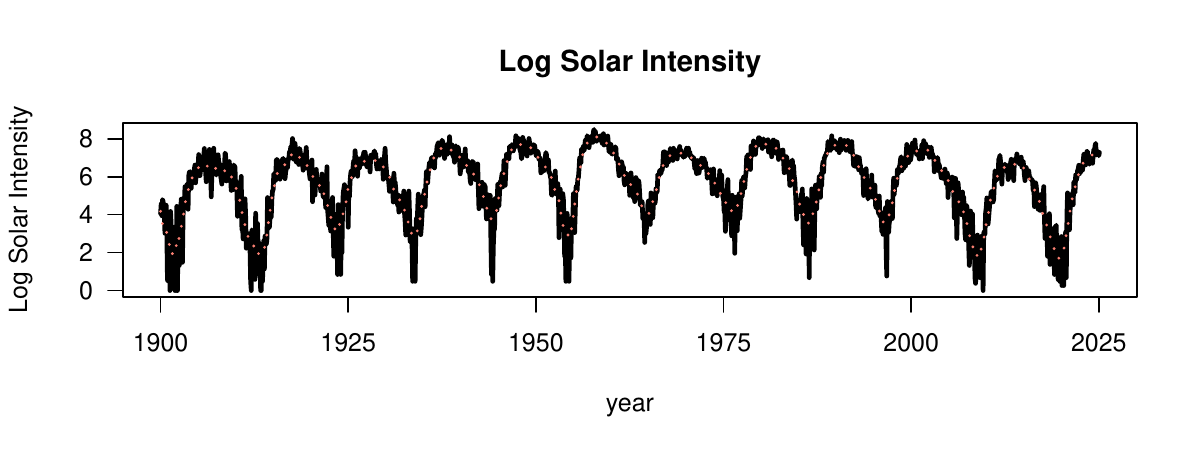}
    \includegraphics[width=0.85\linewidth]{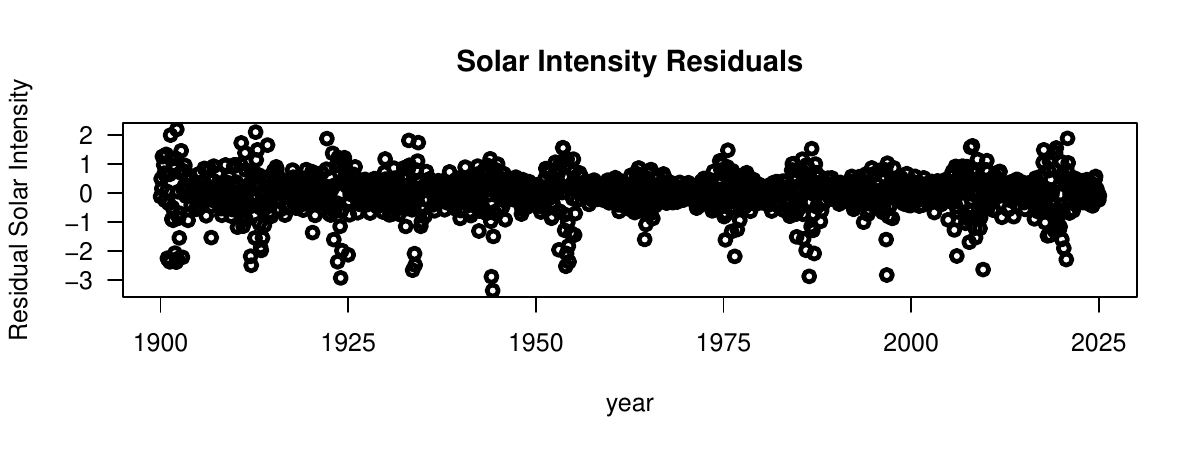}
    \caption{Log solar intensity time series with 
    Nadaraya-Watson kernel regression (top) and
    residuals of the kernel regression (bottom)}
    based on a bandwidth of 18 months.
    \label{fig:solarIntensity}
\end{figure}

To identify autocorrelation in the residual time series
from Figure~\ref{fig:solarIntensity}, we can use both 
our random rotation test (Theorem~\ref{betacon}) and
the Breusch-Godfrey test.  The result is the left 
plot in Figure~\ref{fig:solarAutoCor}.
Unsurprisingly, the Breusch-Godfrey test returns 
highly significant p-values for all lags from 1 to 24,
which merely indicates the strong presence of autocorrelation
in this time series at lag 1 and possibly or not at 
higher lags.  In contrast, our random rotation 
test provides a more nuanced look at the autocorrelation
in this time series.  In particular, we see significant
autocorrelation at lags $h = 1, 6, 7, 8, 9$ at the 0.1\% level.

Going further, we can remove the autocorrelation at lag 1
by fitting an AR(1) model to this residual time series 
via the \texttt{arima} function in R's \texttt{stats} package
and consider the residuals of that AR(1) model.
Running both tests for autocorrelation on
this new residual time series results in the right
plot in Figure~\ref{fig:solarAutoCor}.  The Breusch-Godfrey
test returns significant p-values for all lags 6 and larger at
the 0.1\% level.
In contrast our random rotation test identifies just 
$h = 6$ and $h=9$ as significant lags at the 0.1\% level.

More comprehensively, we also performed our random rotation test for all
lags from 1 to 160 months on the residuals of the AR(1) model
and corrected for multiple testing using classic Benjamini-Hochberg
method.
Table~\ref{tab:topSolarPVs} charts the results.  
In particular, 
the lag of 135 months or 11.25 years has a false discovery rate
of 0.012, 0.013, and 0.06 for the three bandwidths, respectively.
This directly coincides
with the approximate 11 year solar cycles observed by 
astrophysicists.

\begin{table}
    \centering
    \begin{tabular}{lllllllll}
       \hline\hline
       \multicolumn{9}{c}{\bf Bandwidth = 6 months}\\
       Lag  &  \multicolumn{1}{c}{2}  & \multicolumn{1}{c}{3} & \multicolumn{1}{c}{5}  & \multicolumn{1}{c}{135} &  \multicolumn{1}{c}{27} & \multicolumn{1}{c}{133} & \multicolumn{1}{c}{152} & \multicolumn{1}{c}{149} \\
       p-value      &  $6.7\times10^{-39}$ & $9.8\times10^{-29}$ & $0.00022$ 
         & $0.00029$  & $0.0014$ & $0.0031$ & $0.0088$ & $0.018$ \\
         FDR        &  $1.1\times10^{-36}$ & $7.9\times10^{-27}$ & $0.012$ 
         & $0.012$  
         & $0.044$ & $0.084$ & $0.201$ & $0.362$ \\
       \hline
       \multicolumn{9}{c}{\bf Bandwidth = 12 months}\\
       Lag  &  \multicolumn{1}{c}{6}  & \multicolumn{1}{c}{4} & \multicolumn{1}{c}{2}  & \multicolumn{1}{c}{3} &  \multicolumn{1}{c}{5} & \multicolumn{1}{c}{135} & \multicolumn{1}{c}{116} & \multicolumn{1}{c}{103} \\
       p-value      &  $1.4\times10^{-12}$ & $3.7\times10^{-6}$ & $0.000014$ 
         & $0.000063$  & $0.00012$ & $0.00049$ & $0.0045$ & $0.006$ \\
         FDR        &  $2.3\times10^{-10}$ & $0.00030$ & $0.00072$ & $0.0025$  
         & $0.0040$ & $0.0129$ & $0.102$ & $0.120$ \\
       \hline
       \multicolumn{9}{c}{\bf Bandwidth = 18 months}\\
       Lag  &  \multicolumn{1}{c}{6}  & \multicolumn{1}{c}{9} & \multicolumn{1}{c}{135}  & \multicolumn{1}{c}{97} &  \multicolumn{1}{c}{8} & \multicolumn{1}{c}{126} & \multicolumn{1}{c}{116} & \multicolumn{1}{c}{27} \\
       p-value      &  $0.000003$ & $0.0003$ & $0.0011$ & $0.0059$  
         & $0.0059$ & $0.0067$ & $0.0068$ & $0.0095$ \\
         FDR        & 0.00054 & 0.0255 & 0.0595 & 0.1560 & 0.1560 & 0.1560 & 0.1560 & 0.192 \\
       \hline\hline
    \end{tabular}
    \caption{
      The eight smallest p-values from 
      the random rotation test on the residuals of the 
      solar cycle model and their false discovery
      rates corrected via the Benjamini-Hochberg 
      procedure to account for multiple testing 
      are charted for bandwidths of 6, 12, and 18 months.
    }
    \label{tab:topSolarPVs}
\end{table}

\begin{figure}
    \centering
    \includegraphics[width=0.45\linewidth]{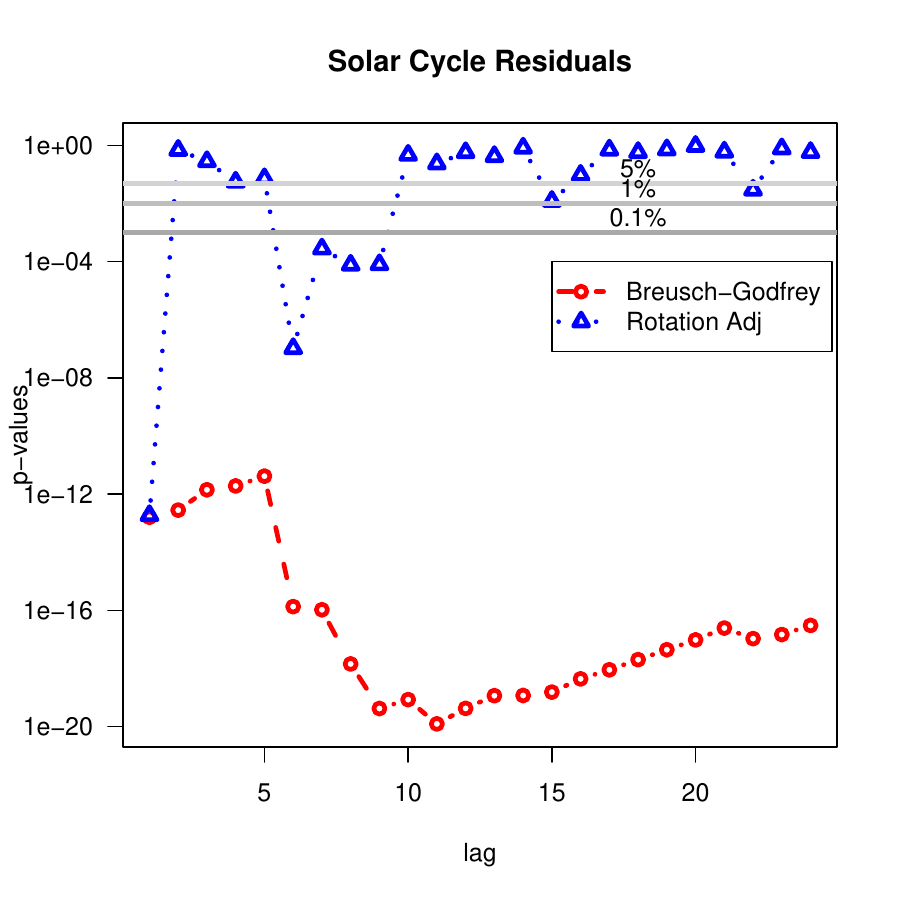}
    \includegraphics[width=0.45\linewidth]{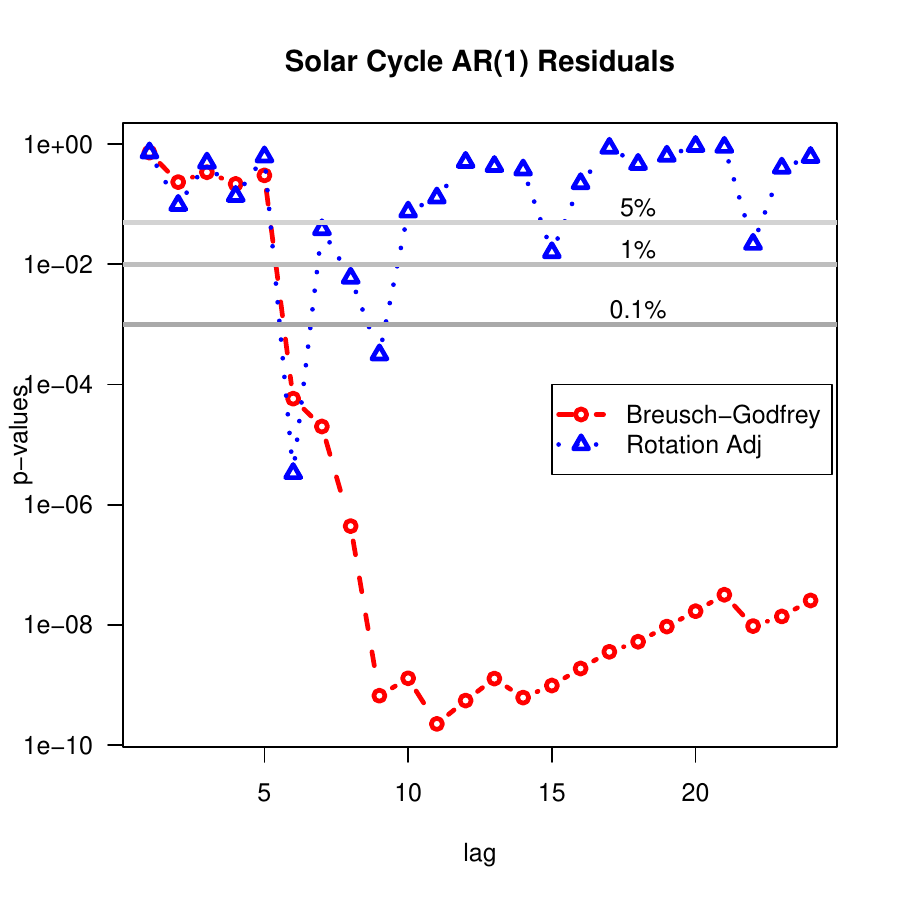}
    \caption{
      A progression of p-values for testing for 
      autocorrelation in the residuals from
      Figure~\ref{fig:solarIntensity} (left) and
      the p-values after removing the AR(1) 
      component (right) both at a bandwidth of 18 months.
    }
    \label{fig:solarAutoCor}
\end{figure}

\section{Summary of Results}

Through the extensive empirical power analysis
carried out in Section~\ref{sec:sims}, we come to 
a couple conclusions.  When there is significant
autocorrelation at lag 1, the classic Durbin-Watson
test has superior statistical power to all other 
methods.  However, when there is significant autocorrelation
at lag $h>1$, our proposed random-rotation test 
(Theorem~\ref{betacon}) yields the strongest statistical
power among the computation-free methods.  
The computation-based randomization tests do achieve
a slight gain in statistical power over our method, but
at the expense of a potentially high computational
burden depending on the size of the dataset under 
consideration.
When there is significant autocorrelation at 
multiple lags,
our random-rotation test run at a single lag offers lower
statistical power than Breusch-Godfrey and Ljung-Box considering all lags.
However, our random-rotation p-values can be combined
using Fisher's method to gain even more statistical power
in the end allowing it to surpass Breusch-Godfrey and Ljung-Box.

For the solar data from Section~\ref{sec:solarData}, 
the Breusch-Godfrey test will always return a significant
p-value as it only can test for significance at  
all lags less than or equal to $h$.  Computation-based
randomization tests can be performed, but are 
slow to run on such a large dataset.  
As our 
random rotation test can test for significant autocorrelation
precisely at lag $h$, it allows us to identify the 
presence of autocorrelation at large lags
with effectively no computation time. 

Beyond autocorrelation in time series data, 
our random rotation approach to hypothesis testing
can be similarly adapted to other statistical tests
in the form of quadratic forms.  This includes future
investigations into many statistics such as those 
for Moran's I for spatial correlation, 
Cramer's test for association in contingency tables, 
the Rayleigh test for spherical uniformity 
in directional data, and many others.

%%%%%%%%%%%%%%%%%%%%%%%%%%%%%%%%%%%%%%%%%%%%%%
%% Single Appendix:                         %%
%%%%%%%%%%%%%%%%%%%%%%%%%%%%%%%%%%%%%%%%%%%%%%
\begin{appendix}

\section{Technical Results and Proofs}
\subsection{Proofs from Section~\ref{sec:quadForm}}
\label{app:proofs1}

\begin{proof}[Proof of Theorem~\ref{thm:lipGen}]
    For a real valued matrix $A$,
    let $B = ({1}/{2})(A + (A\TT{)})$, 
    be the symmetric version of $A$. Then for all 
    $x \in \real^n$, we have $\TT{x}Ax = \TT{x}Bx$.

    Let $g : S^{n-1} \to \real$ be defined by, 
    $g(x) = \TT{x}Ax = \TT{x}Bx$. 
    It is clear that $g$ is continuously differentiable. 
    Since $S^{n-1}$ is compact and any continuously 
    differentiable function over a compact domain is 
    Lipschitz and since the Lipschitz constant is equal 
    to the maximum magnitude of the derivative, 
    we can conclude that $g$ is also Lipschitz. 
    Furthermore, 
    $\nabla g(x) = 2Bx$. 
    Therefore, the Lipschitz constant of $g$ is,
    $
        \max_{x \in S^{n-1}}{\lVert 2Bx \rVert} = 
        2\max_{x \in S^{n-1}}{\lVert Bx \rVert} = 2\rho(B)
    $
    where $\rho(B)$ is the spectral radius of $B$. 
    We can conclude that $g$ is $2\rho(B)$-Lipschitz. 

    Next, let $P, Q \in \mathrm{SO}(n)$ and $x \in S^{n-1}$ 
    be fixed. Then, $F(P) = g(Px)$ and 
    $F(Q) = g(Qx)$. Therefore,
    \begin{align*}
|F(P) - F(Q)| &= |g(Px) - g(Qx)|
\leq 2\rho(B)\lVert Px - Qx \rVert_2\\
%&= 2\lVert (P - Q)x \rVert_2\\
%&\leq 2\lVert (P-Q)\rVert_{op}\lVert x \rVert_2\\
&\leq 2\rho(B)\lVert (P-Q)\rVert_{op}
%&= 2 \sqrt{\lambda_{max}[(P-Q)^t(P-Q)]}\\
%&\leq 2 \sqrt{\sum_{i=1}^n\lambda_{i}[(P-Q)^t(P-Q)]}\\
%&= 2 \sqrt{\text{trace}[(P-Q)^t(P-Q)]}\\
\le 2\rho(B) \lVert (P-Q)\rVert_{HS}
\end{align*}
where $\lVert . \rVert_{HS}$ is the Hilbert-Schmidt norm. 
Therefore, we can conclude that the function 
$F$ is also $2\rho(B)$-Lipschitz with respect to the 
Hilbert-Schmidt metric.
\end{proof}

\begin{proof}[Proof of Theorem \ref{thm:reflectConc}]
  Without loss of generality, let $A$ be symmetric.
  We first note that 
  $F(b) - \xv F(b) = 2\sum_{i> j}^n a_{i,j}b_ib_jx_ix_j$,
  which is a degree 2 homogeneous Rademacher chaos where
  the $b_i$ are iid Rademacher random variables.
  From \cite{KWAPIEN1987}, we have a bound on the $p$th moment
  of $F(b) - \xv F(b)$ in terms of its second moment.  
  In particular,
  $$
    \xv\left[
      \abs{F(b) - \xv F(b)}^p
    \right]
    \le 32^{p/2}\left( \frac{\Gamma(p+1)}{2^{p/2}\Gamma(p/2+1)} \right)
    \left[\xv\abs{F(b) - \xv F(b)}^2\right]^{p/2},
  $$
  and
  \begin{align*}
    \xv\abs{F(b) - \xv F(b)}^2 
    &= 4\xv\left[
      \sum_{i> j,k> l}^n a_{i,j}a_{k,l}b_ib_jb_kb_lx_ix_jx_kx_l
    \right]\\
    &= 4
      \sum_{i>j}^n a_{i,j}^2x_i^2x_j^2
    \le \max_{i>j}\{a_{i,j}^2\}
  \end{align*}
  as $x\in S^{n-1}$.
  These moment bounds can then be translated into tail probability
  bounds via a standard trick involving the moment generating 
  function; see, for example, \cite{BOUCHERON2013}.
  \begin{align*}
    \xv\exp\left( \lmb(F(b)-\xv F(b)) \right)
    &\le 
    \sum_{p=1}^\infty \frac{\lmb^{2p}}{(2p)!}
      \xv [F(b)-\xv F(b)]^{2p}\\
    &\le
    \sum_{p=1}^\infty \frac{\lmb^{2p}}{(2p)!}
    32^{p}\left( \frac{(2p)!}{2^{p}p!} \right)
    \max_{i>j}\{a_{i,j}^{2}\}^p\\
    &=
    \sum_{p=1}^\infty \frac{(16\lmb^{2})^p}{p!}
    \max_{i>j}\{a_{i,j}^{2}\}^p\\
    &\le \exp\left(
      16\lmb^2\max_{i>j}\{a_{i,j}^2\}
    \right).
  \end{align*}
  From Markov's exponential, sometimes referred to as 
  Chernoff's, inequality
  $$
    \probB{
      F(b) - \xv F(b) > t 
    } \le \exp\left(
      -t^2/ 64\max_{i>j}\{a_{i,j}^2\}
    \right).
  $$
\end{proof}

\begin{proof}[Proof of Theorem~\ref{thm:exchangeConc}]
   This theorem essentially follows from the Azuma-Hoeffding inequality
   for martingales with bounded differences
   \citep{AZUMA1967,RIO2013}.  
   Similar to the proof of Theorem~\ref{thm:lipGen},
   we note that $g : S^{n-1} \to \real$ be defined by, 
   $g(x) = \TT{x}Ax = \TT{x}Bx$ with $B = (A + \TT{A})/2$, 
   is Lipschitz. 
   %with constant 
   %$2\rho(B)$ where $\rho(B)$ is the spectral radius of 
   %$B = (A + \TT{A})/2$.
   Indeed, for any
   $u, v \in S^{n-1}$,
   \begin{equation}\label{eqn:liponsphere}
     \abs{g(u) - g(v)}
     = \abs{ \TT{(u - v)} B (u + v) }
     \le \rho(B) \norm{ u - v}\norm{ u + v} 
     \le 2 \rho(B) \norm{ u - v}
   \end{equation}
   using the symmetry of $B$ and $\norm{ u + v} \le 2$.
   Let $F(\pi) = g(\pi x)$.
   
   For 
   $\pi\in \mathbb{S}(n)$, we use the 
   Fisher--Yates construction.
   %We write $\pi$
   %as a product of $n-1$ transpositions 
   %$\pi = \pi_1\cdot\ldots\cdot\pi_{n-1}$, some of 
   %which may be the identity.
   Let $x\in S^{n-1}$ be a fixed vector and 
   $\tau$ be a fixed permutation of $\{1, \dots, n\}$ ordering the
   coordinates of $x$ by decreasing magnitude, i.e.
   $|x_{\tau(1)}| \ge |x_{\tau(2)}| \ge \dots \ge |x_{\tau(n)}|$. 
   We generate a uniformly distributed $\pi$ by first
   letting $V_1$ be uniform on the set $\{1,\ldots,n\}$.
   Then, for $k = 2, \dots, n$, we choose $V_k$ to be 
   uniform on the set $\{1, \dots, n\} \setminus \{V_1, \dots, V_{k-1}\}$.
   Then, $\pi$ is defined by $\pi(V_k) = \tau(k)$
   and $\pi\dist\text{Uniform}(\mathbb{S}_n)$.
   We define the sigma fields 
   $\mathcal F_k = \sigma(V_1, \dots, V_k)$ and
   random variables
   $Z_k = \xv [ F(\pi) \mid \mathcal F_k ]$ for $k = 0, 1, \dots, n$,
   so that $Z_0 = \xv F(\pi)$ and $Z_n = F(\pi)$.
   %note $Z_{n-1} = Z_n$ since the final placement is forced.
   
  Fix $k \le n - 1$, condition on $\mathcal F_{k-1}$, and let $p \ne q$
  be two admissible values for $V_k$.
  Let $\sigma_{pq}$ be the transposition of $p$ and $q$. The two
  sequences continue as follows: given the continuation
  $V_{k+1}, \dots, V_n$ of the run with $V_k = p$, set
$V_j' = \sigma_{pq}(V_j)$ for $j > k$ in the run with $V_k' = q$.
Since $\sigma_{pq}$ maps the set of positions left unoccupied after
step $k$ in the first run bijectively onto the corresponding set in
the second run, and maps uniform choices to uniform choices, this
coupling realizes the correct conditional distribution in both runs.
Under it, the final assignments $\pi$ and $\pi'$ agree except at the
two positions $p$ and $q$: the coordinate $x_{\tau(k)}$ sits at $p$ in
the first run and at $q$ in the second, while the coordinate
$x_{\tau(m)}$ that occupies $q$ in the first run --- necessarily placed
at some later step $m > k$ --- occupies $p$ in the second. Consequently
the vectors $\pi x$ and $\pi' x$ differ in exactly the coordinates $p$
and $q$, and
$$
\norm{ \pi x - \pi' x }^2
 = 2 \left( x_{\tau(k)} - x_{\tau(m)} \right)^2
 \le 2 \left( |x_{\tau(k)}| + |x_{\tau(m)}| \right)^2
 \le 8\, x_{\tau(k)}^2,
$$
where the final inequality uses $m > k$ and the magnitude ordering,
which forces $|x_{\tau(m)}| \le |x_{\tau(k)}|$. Combining with
\eqref{eqn:liponsphere} and taking conditional expectations under the
coupling gives
$$
 \abs*{ \mathbb E[ F(\pi) \mid \mathcal F_{k-1}, V_k = p ]
    - \mathbb E[ F(\pi) \mid \mathcal F_{k-1}, V_k = q ] }
 \le \xv \left[ |F(\pi) - F(\pi')| \right]
 \le 4\sqrt 2\, \rho(B)\, \abs{x_{\tau(k)}} =: \ell_k .
$$
Therefore, conditionally on $\mathcal F_{k-1}$, the map
$v \mapsto \mathbb E[ F(\pi) \mid \mathcal F_{k-1}, V_k = v ]$ has
range at most $\ell_k$, and the martingale increment
$Z_k - Z_{k-1}$ takes values in an interval of length at most
$\ell_k$ with $\mathcal F_{k-1}$-measurable endpoints.
   
   %Then, we let $F(\pi) = g(\pi x)$ where we represent
   %$\pi\in\mathbb{S}(n)$ with the standard $n\times n$ 
   %permutation matrix. Furthermore, we write
   %$F(\pi)=F(\pi_1,\ldots,\pi_n)$, and consider
   %the following difference:
   %\begin{align*}
   %  &\abs{
   %    F(\pi_1,\ldots,\pi_k,\pi_{k+1})-
   %    F(\pi_1,\ldots,\pi_k,\pi_{k+1}')
   %   }\\ 
   %  &~~~~~~= \abs{
   %    g(\pi_1\ldots\pi_k\pi_{k+1} x) - 
   %    g(\pi_1\ldots\pi_k\pi_{k+1}'x)
   %  } \\
   %  &~~~~~~\le 2\rho(B)\norm{(\pi_{k+1} - \pi_{k+1}')x}.
   %\end{align*}
   %Without loss of generality, we assume $\pi$ is
   %a cycle as otherwise we decompose $\pi$ into 
   %disjoint cycles and proceed similarly for each piece.
   %Thus, after relabeling, we can choose each 
   %$\pi_k = (a_k~a_{k+1})$ for distinct $a_k\in\{1,\ldots,n-1\}$.
   %Thus, 
   %$$
   %  \norm{(\pi_{k+1} - \pi_{k+1}')x}^2 = 
   %  (x_{a_k}-x_{a_{k+1}})^2 +
   %  (x_{a_k}-x_{a_{k+1}'})^2 +
   %  (x_{a_{k+1}}-x_{a_{k+1}'})^2.
   %$$
   %Summing over all $k$ gives an upper bound of 6
   %as $x\in S^{n-1}$.  Thus, if we denote $c_k$ to 
   %be the upper bound on 
   %$\abs{
   %    F(\pi_1,\ldots,\pi_k,\pi_{k+1})-
   %    F(\pi_1,\ldots,\pi_k,\pi_{k+1}')
   %   }$, then $\sum_{i=k}^n c_k^2 = 24\rho(B)^2$.
   By the Azuma--Hoeffding inequality for martingales with increments
   of lengths bounded by $\ell_k$---see also \cite{mcdiarmid1989},
  $$
    \prob{ F(\pi) - \mathbb E F(\pi) > t }
    \le \exp \left\{ 
      - \frac{2 t^2}{\sum_{k=1}^{n-1} \ell_k^2} 
    \right\},
    \quad
    \sum_{k=1}^{n-1} \ell_k^2
    \le 32\, \rho(B)^2 \sum_{k=1}^{n} x_{\tau(k)}^2
    = 32\, \rho(B)^2 ,
  $$
  which yields the stated bound.
   %This theorem follows from concentration of normalized
   %counting measure on finite metric spaces; 
   %see \cite{ledoux} Theorem~{4.2} and 
   %Corollary~{4.3}.  
   %For two permutations $\pi$ and $\sigma$, 
   %the Hamming metric is defined as
   %$$
   %  d(\pi,\sigma) = \abs{\{
   %    i \,:\, \pi(i) \ne \sigma(i)
   %  \}}.
   %$$
  \iffalse
   Following from the proof of Lemma~\ref{lem:lipGen},
   we note that $g : S^{n-1} \to \real$ be defined by, 
   $g(x) = \TT{x}Ax$, is Lipschitz with constant 
   $2\rho(B)$ where $\rho(B)$ is the spectral radius of 
   $B = (A + \TT{A})/2$.
   Then, we let $F(\pi) = g(\pi x)$ where we represent
   $\pi\in\mathbb{S}(n)$ with the standard $n\times n$ 
   permutation matrix. This results in 
   \begin{align*}
     \abs{F(\pi)-F(\sigma)}
     &= \abs{g(\pi x) - g(\sigma x)} \\
     &\le 2\rho(B)\norm{(\pi - \sigma) x}\\
     &\le 2\rho(B)\norm{ \pi - \sigma }_\text{OP}
     \le 4\rho(B)
   \end{align*}
   From concentration on finite metric spaces
   \citep[Theorem 4.2]{ledoux}, we conclude that
   $$
     \probB{
       F(\pi) - \xv F(\pi) > t 
     } \le \exp\left(
       -t^2/ 64n\rho(B)^2
     \right).
   $$
  \fi
\end{proof}

\subsection{Proofs from Section~\ref{sec:randRot}}
\label{app:proofs2}

\begin{proof}[Proof of Corollary \ref{cor:powerAnal}]
  Let $
    pv_\vee(X) := 
    \exp{\{{-(n-2)(\hat{T}_{h,X}(I))^2}/{64}\}}
  $
  be a random p-value
  under $\mathcal{G} = SO(n)$ with random vector
  $X$ having autocorrelation
  $\eta_h = \delta n^{-1/2}$. Then, 
  denoting $\hat{T}_{h,X}(I)$ as just $\hat{T}$,
  \begin{equation}
    \label{eqn:pvTerms}
    pv_\vee(X) =
    \exp\left\{
    -\frac{n-2}{64}\left[ 
      ( \hat{T} - \xv\hat{T} )^2 -
      2( \hat{T} - \xv\hat{T} )\xv\hat{T} +
      (\xv\hat{T})^2
    \right]
    \right\}.
  \end{equation}
  Because the $X_t$ are such that $\norm{X}_2^2=1$ and
  the process is stationary, $\var{X_t}=1/n$ for all $t$.
  As
  $\hat{T}$ is an unnormalized summation,
  the $\sqrt{n}$-consistency of the sample autocorrelation
  implies that 
  $$
    \sqrt{n}\left[\frac{1}{n-h}\sum_{t=1}^{n-h}\frac{X_tX_{t+h}}{1/n}
    - \eta_h
    \right] \convd Z, \text{a mean zero Gaussian random variable.}
  $$
  And thus, $\hat{T} = \sum_{t=1}^{n-h}{X_tX_{t+h}}$
  also converges at a $\sqrt{n}$-rate to the autocorrelation
  $\eta_h$.
  
  Considering the three terms above in 
  Equation~\ref{eqn:pvTerms}, 
  the first
  $\exp\{
    -\frac{n-2}{64} 
      ( \hat{T} - \xv\hat{T} )^2\} \le 1$.
  In particular, $\sqrt{n}(\hat{T} - \xv\hat{T})$
  converges in distribution to a Gaussian random 
  variable by the central limit theorem.
  By the law of 
  large numbers, we have that 
  $( \hat{T} - \xv\hat{T} )\rightarrow 0$ almost surely
  at the usual $\sqrt{n}$-rate, 
  and thus by the continuous mapping theorem 
  $(n-2)( \hat{T} - \xv\hat{T} )\xv \hat{T} = O(1)$.
  Lastly, we have $(n-2)[\xv\hat{T}]^2 = O(1)$ by assumption.
  This results in the p-value converging in distribution to a random variable
  on the unit interval.

  If instead $\eta_h = \delta n^{-1/2+\veps}$ for some 
  $\veps>0$, then
  $$
    \log pv_\vee(X) = O(n^\veps-n^{2\veps})
  $$
  and once again by the continuous mapping theorem
  and the law of large numbers, we have that 
  $pv_\vee(X)\convas 0$ as $n\rightarrow\infty$.
  %from Jensen's inequality,
  %$$
  %  \xv \log pv_\vee(X) \le
  %  -\frac{n-2}{64}[\xv \hat{T}_{h,X}(I)]^2 =
  %  -\frac{\delta^2}{64}\frac{n-2}{n^{1-2\veps}}
  %  = O(-n^{2\veps}).
  %$$
\end{proof}

\begin{proof}[Proof of Proposition~\ref{prop:exact}]
The law of $z / \norm{z}$ is invariant under any rotation, and
the uniform distribution is the unique rotation-invariant probability
measure on $S^{n-1}$. Since $M' (Mx) \eqdist Mx$ for every fixed
$M' \in SO(n)$, the law of $Mx$ is also rotation invariant, whence
$Mx \eqdist z / \norm{z}$. Therefore
$$
  \hat T_{h,x}(M) \eqdist \frac{z^{\mathrm T} B_h z}{z^{\mathrm T} z},
  \text{ so }
  \prob{ \hat T_{h,x}(M) \ge t } = 
  \prob{ z^{\mathrm T} (B_h - t I_n) z \ge 0 },
$$
and diagonalizing $B_h - t I_n = U \,\mathrm{diag}(\lambda_i - t)\, U^{\mathrm T}$
with $U$ orthogonal, the vector $U^{\mathrm T} z$ is again standard
Gaussian, which yields the stated representation.
 
For the spectrum, partition $\{1, \dots, n\}$ into the $h$ residue
classes modulo $h$. The matrix $B_h$ couples index $i$ only to
$i \pm h$, so after a simultaneous permutation of rows and columns that
groups indices by residue class, $B_h$ is block diagonal, each block
being one half of the adjacency matrix of a path graph: $r$ paths on
$q + 1$ vertices (residues $1, \dots, r$) and $h - r$ paths on $q$
vertices. The adjacency matrix of the path on $m$ vertices has
eigenvalues $2\cos( j\pi / (m+1) )$, $j = 1, \dots, m$; see also
\cite{toep} for the general pseudo-Toeplitz setting. As a
check, $\sum_i \lambda_i^2 = (n-h)/2 = \operatorname{tr}(B_h^2)$,
consistent with Lemma~3.7.
\end{proof}

\begin{proof}[Proof of Lemma \ref{lem:nullVar}]
    We start by seeing that, 
    $\text{Var}(\hat{T}_{h,X}(I)) = 
    \mathrm{E}[\hat{T}_{h,X}(I)^2]$, 
    since $\mathrm{E}[\hat{T}_{h,X}(I)] = 0$ 
    under $H_0$. 
    Next, 
    $\hat{T}_{h,X}(I)^2 = 
    (\TT{X}B_hX)^2 = \TT{X}B_hX  \TT{X}B_hX$ 
    where $X$ is uniformly distributed in the sphere $S^{n-1}$. 
    Using Lemma \ref{var} and that the trace of $B_h$ is zero, 
    we have
    $
        \mathrm{E}[\hat{T}_{h,X}(I)^2] = 
        \mathrm{E}[\TT{X}B_hX  \TT{X}B_hX] = 
        \mu_{22}(2\text{tr}(B_h^2))
    $
    with $\mu_{22} = \mathrm{E}[X_1^2X_2^2]$. 
    To calculate $\mu_{22}$, we see that
    $$
        X_1^2 X_2^2 = (1 - X_2^2 - \cdots - X_n^2) X_2^2 = X_2^2 - X_2^4 - X_3^2 X_2^2 - \cdots - X_n^2 X_2^2.
    $$
    By symmetry, $\mu_{22} = \mathrm{E}[X_i^2X_j^2]$ for 
    any $i \neq j$. Therefore, taking expectations gives
    $
        \mu_{22} = \mathrm{E}[X_2^2] - \mathrm{E}[X_2^4] - (n-2) \mu_{22}
    $
    Upon rearranging, we get
    $
        \mu_{22} = 
        \left( 
          \mathrm{E}[X_2^2] - \mathrm{E}[X_2^4] 
        \right)/(n-1).
    $
    As $\sum_{1}^n X_i^2 = 1$ almost surely, 
    taking expectation gives 
    $\mathrm{E}[X_2^2] = {n}^{-1}$ by symmetry. 
    Also, since 
    $X_2^2 \dist
    \text{Beta}({1}/{2}, {(n-1)}/{2})$, $\mathrm{E}[X_2^4]$ 
    is equal to the second moment of the beta distribution, which is 
    ${3}/{n(n+2)}$. Therefore,
    \begin{equation}
        \mu_{22} = 
        \frac{1}{n-1}\left( \frac{1}{n} - \frac{3}{n(n+2)} \right) 
        = \frac{1}{n(n+2)} \:.
    \end{equation}
    Lastly, we calculate the trace of $B_h^2$.  Since
    %\begin{multline*}
    $$
      B_h^2 
      = \frac{1}{4}\left(
          A^{2h} + \TT{(A^{2h})} + A^h\TT{(A^h)} + \TT{(A^h)}A^h
        \right)\\
      = \frac{1}{2}B_{2h} + 
        \frac{1}{4}\begin{pmatrix}
          I_{n-h} & 0\\
          0 & 0
      \end{pmatrix} + \frac{1}{4}\begin{pmatrix}
          0 & 0\\
          0 & I_{n-h}
      \end{pmatrix},
    $$
    %\end{multline*}
    $\tr{B_h^2} = (n-h)/2$.
    Thus, finally we have our variance,
    $
        \text{var}(\hat{\phi}) = {(n-h)}/{n(n+2)}.
    $
\end{proof}

\begin{proof}[Proof of Theorem~\ref{betacon}]
    Under $H_0$, consider $\veps_t$ to be a stationary sequence
    of independent mean zero random variables such 
    that $\veps_t\veps_{t+h}$ has
    unit variance and finite third moment.
    The terms in the summation 
    $S_{n}=\sum_{t=1}^{n-h} \veps_t\veps_{t+h}$
    are independent when the index is greater than 
    the lag; i.e. 
    $\veps_t\veps_{t+h} \ind \veps_{t+i}\veps_{t+h+i}$ for 
    $i>h$.
    This satisfies the notion of $h$-independence (referred
    to as $m$-independence) from \cite{TIKHOMIROV1981}.
    Theorem 5 of \cite{TIKHOMIROV1981} gives a Berry-Esseen-type
    bound for the Kolmogorov distance for such 
    $h$-independent sequences.
    Hence,
    for $F_{n-h}$ the empirical distribution function of  
    $S_{n}$ and denoting the standard Gaussian distribution
    function as $\Phi(z)$,
    we have that 
    $$
      \sup_z \abs{ F_{n-h}(z) - \Phi(z) } 
      \le C_1 (n-h)^{-1/2} + C_2 (n-h)^{-1}\log (n-h)
      = O(n^{-1/2})
    $$
    for constants $C_1,C_2>0$ dependent on the variance
    and third moment but independent of $n$.
    By a simple derivation and $W$ a standard normal
    random variable, we have that
    \begin{multline*}
      \abs*{ \probB{S_n^2 < t} - \probB{W^2<t}  }
      =
      \abs*{ \probB{\abs{S_n} < {t}^{1/2}} - \probB{\abs{W}<{t}^{1/2}} }\\
      \le
      \abs*{ \probB{S_n < {t}^{1/2}} - \probB{W < {t}^{1/2}} } 
      +
      \abs*{ \probB{S_n < -{t}^{1/2}} - \probB{W < -{t}^{1/2}} }
      = O(n^{-1/2}),
    \end{multline*}
    which gives a proximity between the empirical distribution
    of $S_n$ and the $\distChiSquared{1}$ distribution.

    Following from Lemma~\ref{lem:nullVar}, let
    $$
      Z:= \frac{\hat{T}_{h,X}(I)^2}{(n-h)/n(n+2)}
      ~~\text{ and }~~
      c:= \frac{64/(n-2)}{(n-h)/n(n+2)}
    $$
    so that $\xv Z = 1$ and $Z/c = (n-2)\hat{T}_{h,X}(I)^2/64$.
    From the above, 
    the distribution of $Z$ is close to 
    $\chi^2(1)$ for large $n$.
    Thus, for some $c > 0$ and $u \in (0,1)$, we have,
    \begin{align*}
        \probB{e^{-Z/c} \leq u} &= \probB{Z \geq -c\log{u}}
        \leq ({2\pi})^{-1/2}\int_{-c\log{u}}^{\infty}x^{-1/2}e^{-x/2} dx + O(n^{-1/2})\\
        &= \left({\frac{c}{2\pi}}\right)^{1/2}\int_{0}^u(-\log{y})^{-1/2}y^{c/2-1} dy + O(n^{-1/2})\\
        &\leq \left({\frac{c}{2\pi}}\right)^{1/2}\int_{0}^u(1-y)^{-1/2}y^{c/2-1} dy + O(n^{-1/2})\\
        &= \frac{(c/2)^{1/2}\Gamma(c/2)}{\Gamma((c+1)/2)} I(u;c/2,1/2) + O(n^{-1/2})
    \end{align*}
    where the inequality $-\log{y} \geq 1-y$ for $y \in (0,1)$ is used. 
    Inputting the value of $c$ from above, we get that,
    $$
        \mathbb{P}\left[\exp{\bigg\{\frac{-(n-2)\hat{T}_{h,X}(I)^2}{64}\bigg\}} \leq u \right] \leq C_0 I\left[u;\frac{32n(n+2)}{(n-h)(n-2)},\frac{1}{2}\right]
         + O(n^{-1/2})
    $$
    where $C_0 = \left(\frac{32n(n+2)}{(n-h)(n-2)}\right)^{1/2}\Gamma \left(\frac{32n(n+2)}{(n-h)(n-2)}\right) \Gamma \left(\frac{1}{2} + \frac{32n(n+2)}{(n-h)(n-2)}\right)^{-1}$.

\end{proof}

%\subsection{Proofs from Section~\ref{sec:multiTest}}

\begin{proof}[Proof of Proposition~\ref{prop:uncorr}]
  Following from Proposition~\ref{lem:rotInnerProduct},
  let $B^{h} := \frac{1}{2}(A^h + \TT{(A^h)})$.
  As $\xv T_1 = \xv T_2 = 0$,
  Lemma~\ref{var} shows that 
  \begin{multline*}
  \cov{ \hat{T}_{h_1,X}(I) }{ \hat{T}_{h_2,X}(I) }
  = \xv\left(
    \TT{X}A^{h_1}X  \TT{X}A^{h_2}X
  \right)
  = \xv\left(
    \TT{X}B^{h_1}X  \TT{X}B^{h_2}X
  \right)\\
  = \mu_{22}[\mathrm{tr}(B^{h_1})\mathrm{tr}(B^{h_2}) + 
    2\mathrm{tr}(B^{h_1+h_2})]
    +
    (\mu_4 - 3\mu_{22})\tr{ B^{h_1}\circ B^{h_2} }
  = 0.
  \end{multline*}
\end{proof}

\subsection{Technical Results}
\label{app:technical}

\begin{proposition}
  \label{lem:rotInnerProduct}
  Let
  $A\in\real^{n\times n}$ be a symmetric
  matrix with spectrum $\lmb_1\ge\lmb_2\ge\ldots\ge\lmb_n$, 
  and let
  $b_A:\real^n\times\real^n\rightarrow \real$
  be a bilinear form defined by 
  $$
    b_{A,SO(n)}(x,y) = 
    \int_{SO(n)}
    \TT{(Mx)}A(My) d\rho(M)
  $$
  where integration is taken over $SO(n)$ with 
  respect to Haar measure $\rho$.
  Then, $b_{A,SO(n)}(x,y)$ is rotationally invariant
  and furthermore
  $
    b_{A,SO(n)}(x,y) = \bar{\lmb}\iprod{x}{y}
  $
  where $\iprod{\cdot}{\cdot}$ is the standard Euclidean 
  inner product and $\bar{\lmb} = n^{-1}\sum_{i=1}^n\lmb_i$. 
\end{proposition}

\begin{proof}
  %By ``Weyl's trick'' \citep[Theorem 2.10]{HOFMANNMORRIS},
  %$b$ is an inner product relative to which all $M\in SO(n)$ 
  %are unitary operators.  
  %In general, 
  Any bilinear form on 
  a real Hilbert Space is of the form $\iprod{Mx}{y}$ for
  some bounded operator $M$.
  Hence, $b(x,y) = \sum_{i=1}^ncx_iy_i$ for
  some $c>0$ by rotational invariance.  
  Without loss of generality, let $A$ be diagonal 
  with entries $\lmb_1,\ldots,\lmb_n$.  Then, choosing $x=y$ to
  be any unit vector results in 
  $$
    c = \int_{\norm{v}=1} \TT{v}Av dv
      = \int_{\norm{v}=1} \sum_{i=1}^n \lmb_i v_i^2 d\mu
  $$
  where $\mu$ is the uniform surface measure of the $(n-1)$-sphere. 
  By symmetry, the integral can be restricted to a fraction 
  of the sphere where $v_i\ge0$.  Furthermore, $\{v_i=v_j\,:\,i\ne j\}$
  is a measure zero event.  Thus,
  \begin{align*}
    c &= 2^n
    \int_{v_1>\ldots>v_n\ge0} \sum_{\pi\in\mathbb{S}_n} 
      \sum_{i=1}^n \lmb_i v_{\pi(i)}^2 d\mu
    = 2^n\int_{v_1>\ldots>v_n\ge0}  
       \sum_{i=1}^n \lmb_{i} \sum_{\pi\in\mathbb{S}_n} v_{\pi(i)}^2  d\mu\\
    &= 2^n\int_{v_1>\ldots>v_n\ge0}  
       \sum_{i=1}^n \lmb_{i}(n-1)!  d\mu
    = \frac{2^n}{n!2^n}\sum_{i=1}^n \lmb_{i}(n-1)! = 
    \frac{1}{n}\sum_{i=1}^n \lmb_{i}
  \end{align*}
  as the sum is over $n!$ permutations in $\mathbb{S}_n$, 
  which is grouped into $(n-1)!$ sets of $n$ $v_i^2$'s 
  that sum to 1.
\end{proof}

\end{appendix}

%%%%%%%%%%%%%%%%%%%%%%%%%%%%%%%%%%%%%%%%%%%%%%
%% Acknowledgements                         %%
%% should be provided in the                %%
%% Acknowledgements section.                %%
%%%%%%%%%%%%%%%%%%%%%%%%%%%%%%%%%%%%%%%%%%%%%%

\section*{Acknowledgments}
The authors would like to thank the anonymous referees, an Associate
Editor and the Editor for their constructive comments that improved the
quality of this paper.

%Information, such as contract numbers, of no interest to readers, should be excluded.

%%%%%%%%%%%%%%%%%%%%%%%%%%%%%%%%%%%%%%%%%%%%%%
%% Funding information, if any,             %%
%% should be provided in the                %%
%% funding section.                         %%
%%%%%%%%%%%%%%%%%%%%%%%%%%%%%%%%%%%%%%%%%%%%%%
\section*{Funding}
The authors would like to thank the 
Natural Sciences and Engineering Research Council of Canada
for their funding support for this research.

%%%%%%%%%%%%%%%%%%%%%%%%%%%%%%%%%%%%%%%%%%%%%%
%% Supplementary Material, including data   %%
%% sets and code, should be provided in     %%
%% {supplement} environment with title      %%
%% and short description. It cannot be      %%
%% available exclusively as external link.  %%
%% All Supplementary Material must be       %%
%% available to the reader on Project       %%
%% Euclid with the published article.       %%
%%%%%%%%%%%%%%%%%%%%%%%%%%%%%%%%%%%%%%%%%%%%%%
\section*{R Code}
%\stitle{Title of Supplement A}
%\sdescription{Short description of Supplement A.}
%\end{supplement}
%\begin{supplement}
%\stitle{Title of Supplement B}
%\sdescription{Short description of Supplement B.}
\label{SM}
R code to recreate all of the simulation experiments
and the analysis of the solar cycle data can be 
found at 
\url{https://github.com/cachelack/Random-Rotation-Tests.git}.

%%  MR numbers will be added by VTeX.                      %%
%%                                                         %%
%%  Use \cite{...} to cite references in text.             %%
%%                                                         %%
%%%%%%%%%%%%%%%%%%%%%%%%%%%%%%%%%%%%%%%%%%%%%%%%%%%%%%%%%%%%%

\bibliographystyle{plainnat}
\bibliography{kasharticle,kashbook,kashself,kashpack,amit}

\begin{thebibliography}{41}
\providecommand{\natexlab}[1]{#1}
\providecommand{\url}[1]{\texttt{#1}}
\expandafter\ifx\csname urlstyle\endcsname\relax
  \providecommand{\doi}[1]{doi: #1}\else
  \providecommand{\doi}{doi: \begingroup \urlstyle{rm}\Url}\fi

\bibitem[Azuma(1967)]{AZUMA1967}
Kazuoki Azuma.
\newblock Weighted sums of certain dependent random variables.
\newblock \emph{Tohoku Mathematical Journal, Second Series}, 19\penalty0
  (3):\penalty0 357--367, 1967.

\bibitem[Boucheron et~al.(2013)Boucheron, Lugosi, and Massart]{BOUCHERON2013}
St{\'e}phane Boucheron, G{\'a}bor Lugosi, and Pascal Massart.
\newblock \emph{Concentration inequalities: A nonasymptotic theory of
  independence}.
\newblock Oxford University Press, 2013.

\bibitem[Box and Pierce(1970)]{BOXPIERCE1970}
George~EP Box and David~A Pierce.
\newblock Distribution of residual autocorrelations in
  autoregressive-integrated moving average time series models.
\newblock \emph{Journal of the American statistical Association}, 65\penalty0
  (332):\penalty0 1509--1526, 1970.

\bibitem[Breusch(1978)]{BREUSCH1978}
Trevor~S Breusch.
\newblock Testing for autocorrelation in dynamic linear models.
\newblock \emph{Australian economic papers}, 17\penalty0 (31), 1978.

\bibitem[Brockwell and Davis(2009)]{brockwell2009}
Peter~J Brockwell and Richard~A Davis.
\newblock \emph{Time series: theory and methods}.
\newblock Springer science \& business media, 2009.

\bibitem[Davies(1980)]{davies1980}
Robert~B Davies.
\newblock The distribution of a linear combination of $\chi^2$ random
  variables.
\newblock \emph{Journal of the Royal Statistical Society Series C: Applied
  Statistics}, 29\penalty0 (3):\penalty0 323--333, 1980.

\bibitem[Davis and Resnick(1986)]{DAVISRESNICK1986}
Richard Davis and Sidney Resnick.
\newblock Limit theory for the sample covariance and correlation functions of
  moving averages.
\newblock \emph{The Annals of Statistics}, pages 533--558, 1986.

\bibitem[Dobriban(2022)]{DOBRIBAN2022}
Edgar Dobriban.
\newblock Consistency of invariance-based randomization tests.
\newblock \emph{The Annals of Statistics}, 50\penalty0 (4):\penalty0
  2443--2466, 2022.

\bibitem[Duchesne and Lafaye~de Micheaux(2010)]{duchesne2010}
Pierre Duchesne and Pierre Lafaye~de Micheaux.
\newblock Computing the distribution of quadratic forms: Further comparisons
  between the {Liu--Tang--Zhang} approximation and exact methods.
\newblock \emph{Computational Statistics \& Data Analysis}, 54\penalty0
  (4):\penalty0 858--862, 2010.

\bibitem[Durbin and Watson(1950)]{DurbinWatsonA}
James Durbin and Geoffrey~S Watson.
\newblock Testing for serial correlation in least squares regression: {I}.
\newblock \emph{Biometrika}, 37\penalty0 (3/4):\penalty0 409--428, 1950.

\bibitem[Durbin and Watson(1951)]{DurbinWatsonB}
James Durbin and Geoffrey~S Watson.
\newblock Testing for serial correlation in least squares regression: {II}.
\newblock \emph{Biometrika}, 38\penalty0 (1-2):\penalty0 159--178, 1951.

\bibitem[Durbin and Watson(1971)]{DurbinWatsonC}
James Durbin and Geoffrey~S Watson.
\newblock Testing for serial correlation in least squares regression. {III}.
\newblock \emph{Biometrika}, 58\penalty0 (1):\penalty0 1--19, 1971.

\bibitem[Finos(2025)]{FLIP}
Livio Finos.
\newblock \emph{flip: Multivariate Permutation Tests}, 2025.
\newblock URL \url{https://CRAN.R-project.org/package=flip}.
\newblock R package version 2.5.1.

\bibitem[Godfrey(1978)]{GODFREY1978}
Leslie~G Godfrey.
\newblock Testing against general autoregressive and moving average error
  models when the regressors include lagged dependent variables.
\newblock \emph{Econometrica: Journal of the Econometric Society}, pages
  1293--1301, 1978.

\bibitem[Good(2013)]{GOOD2013}
Phillip Good.
\newblock \emph{Permutation tests: a practical guide to resampling methods for
  testing hypotheses}.
\newblock Springer Science \& Business Media, 2013.

\bibitem[Hoffmann and Kranz(2024)]{DESK}
Soenke Hoffmann and Tobias Kranz.
\newblock \emph{desk: Didactic Econometrics Starter Kit}, 2024.
\newblock URL \url{https://CRAN.R-project.org/package=desk}.
\newblock R package version 1.1.2.

\bibitem[Imhof(1961)]{IMHOF1961}
Jean-Pierre Imhof.
\newblock Computing the distribution of quadratic forms in normal variables.
\newblock \emph{Biometrika}, 48\penalty0 (3/4):\penalty0 419--426, 1961.

\bibitem[Inder(1986)]{INDER1986}
Brett Inder.
\newblock An approximation to the null distribution of the {D}urbin-{W}atson
  statistic in models containing lagged dependent variables.
\newblock \emph{Econometric Theory}, 2\penalty0 (3):\penalty0 413--428, 1986.

\bibitem[Kallenberg(2006)]{KALLENBERG2006}
Olav Kallenberg.
\newblock \emph{Probabilistic symmetries and invariance principles}.
\newblock Springer Science \& Business Media, 2006.

\bibitem[Kashlak and Yuan(2022)]{KASHLAK_YUAN_ABELECT}
Adam~B Kashlak and Weicong Yuan.
\newblock Computation-free nonparametric testing for local spatial association
  with application to the {US} and {C}anadian electorate.
\newblock \emph{Spatial Statistics}, 48:\penalty0 100617, 2022.

\bibitem[Kashlak et~al.(2022)Kashlak, Myroshnychenko, and
  Spektor]{KASHLAK_KHINTCHINE2020}
Adam~B Kashlak, Sergii Myroshnychenko, and Susanna Spektor.
\newblock Analytic permutation testing for functional data {ANOVA}.
\newblock \emph{Journal of Computational and Graphical Statistics}, pages
  1--10, 2022.

\bibitem[King and Wu(1991)]{KING1991}
Maxwell~L King and Ping~X Wu.
\newblock Small-disturbance asymptotics and the {D}urbin-{W}atson and related
  tests in the dynamic regression model.
\newblock \emph{Journal of Econometrics}, 47\penalty0 (1):\penalty0 145--152,
  1991.

\bibitem[Koning(2024)]{KONING2024}
Nick~W Koning.
\newblock More power by using fewer permutations.
\newblock \emph{Biometrika}, 111\penalty0 (4):\penalty0 1405--1412, 2024.

\bibitem[Koning and Hemerik(2023)]{KONING_HEMERIK_2023}
Nick~W Koning and Jesse Hemerik.
\newblock {More Efficient Exact Group Invariance Testing: using a
  Representative Subgroup}.
\newblock \emph{Biometrika}, page asad050, 09 2023.
\newblock ISSN 1464-3510.
\newblock \doi{10.1093/biomet/asad050}.
\newblock URL \url{https://doi.org/10.1093/biomet/asad050}.

\bibitem[Kulkarni et~al.(1999)Kulkarni, Schmidt, and Tsui]{toep}
Devadatta Kulkarni, Darrell Schmidt, and Sze-Kai Tsui.
\newblock Eigenvalues of tridiagonal pseudo-toeplitz matrices.
\newblock \emph{Linear Algebra and its Applications}, 297\penalty0
  (1):\penalty0 63--80, 1999.
\newblock ISSN 0024-3795.
\newblock \doi{https://doi.org/10.1016/S0024-3795(99)00114-7}.
\newblock URL
  \url{https://www.sciencedirect.com/science/article/pii/S0024379599001147}.

\bibitem[Kuonen(1999)]{kuonen1999}
Diego Kuonen.
\newblock Saddlepoint approximations for distributions of quadratic forms in
  normal variables.
\newblock \emph{Biometrika}, 86\penalty0 (4):\penalty0 929--935, 1999.

\bibitem[Kwapien(1987)]{KWAPIEN1987}
Stanislaw Kwapien.
\newblock Decoupling inequalities for polynomial chaos.
\newblock \emph{The Annals of Probability}, 15\penalty0 (3):\penalty0
  1062--1071, 1987.

\bibitem[Langsrud(2005)]{LANGSRUD2005}
{\O}yvind Langsrud.
\newblock Rotation tests.
\newblock \emph{Statistics and computing}, 15\penalty0 (1):\penalty0 53--60,
  2005.

\bibitem[Ledoux(2001)]{ledoux}
M.~Ledoux.
\newblock \emph{The Concentration of Measure Phenomenon}.
\newblock Mathematical surveys and monographs. American Mathematical Society,
  2001.
\newblock ISBN 9780821837924.
\newblock URL \url{https://books.google.ca/books?id=mCX_cWL6rqwC}.

\bibitem[Lehmann and Romano(2006)]{LEHMANN2006}
Erich~L Lehmann and Joseph~P Romano.
\newblock \emph{Testing statistical hypotheses}.
\newblock Springer Science \& Business Media, 2006.

\bibitem[Ljung and Box(1978)]{LJUNGBOX1978}
Greta~M Ljung and George~EP Box.
\newblock On a measure of lack of fit in time series models.
\newblock \emph{Biometrika}, 65\penalty0 (2):\penalty0 297--303, 1978.

\bibitem[McDiarmid(1989)]{mcdiarmid1989}
Colin McDiarmid.
\newblock On the method of bounded differences.
\newblock \emph{Surveys in combinatorics}, 141\penalty0 (1):\penalty0 148--188,
  1989.

\bibitem[Meckes(2019)]{MECKES2019}
Elizabeth~S Meckes.
\newblock \emph{The random matrix theory of the classical compact groups},
  volume 218.
\newblock Cambridge University Press, 2019.

\bibitem[Ong et~al.(2026)Ong, Choi, Lundy, Zhang, Wan, Xu, Ghosh, Jiang, and
  Shi]{ONG2026}
Vy~Q Ong, Bich~N Choi, Devin~P Lundy, Yingnan Zhang, Yin Wan, Hongyan Xu, Santu
  Ghosh, Hui Jiang, and Yang Shi.
\newblock Mcmc-ce: A novel and efficient algorithm for estimating small
  right-tail probabilities of quadratic forms with applications in genomics.
\newblock \emph{Journal of Computational Biology}, 33\penalty0 (2):\penalty0
  236--254, 2026.

\bibitem[Pesarin and Salmaso(2010)]{PESARIN2010}
Fortunato Pesarin and Luigi Salmaso.
\newblock \emph{Permutation tests for complex data: theory, applications and
  software}.
\newblock John Wiley \& Sons, 2010.

\bibitem[Rio(2013)]{RIO2013}
Emmanuel Rio.
\newblock Extensions of the {H}oeffding-{A}zuma inequalities.
\newblock \emph{Electronic Communications in Probability}, 18\penalty0 (54),
  2013.

\bibitem[Shumway and Stoffer(2000)]{SHUMWAY2000}
Robert~H Shumway and David~S Stoffer.
\newblock \emph{Time series analysis and its applications}, volume~3.
\newblock Springer, 2000.

\bibitem[Solari et~al.(2014)Solari, Finos, and Goeman]{SOLARI2014}
Aldo Solari, Livio Finos, and Jelle~J Goeman.
\newblock Rotation-based multiple testing in the multivariate linear model.
\newblock \emph{Biometrics}, 70\penalty0 (4):\penalty0 954--961, 2014.

\bibitem[Tikhomirov(1981)]{TIKHOMIROV1981}
Alexander~N Tikhomirov.
\newblock On the convergence rate in the central limit theorem for weakly
  dependent random variables.
\newblock \emph{Theory of Probability \& Its Applications}, 25\penalty0
  (4):\penalty0 790--809, 1981.

\bibitem[Wiens(1992)]{wiens}
Douglas~P Wiens.
\newblock On moments of quadratic forms in non-spherically distributed
  variables.
\newblock \emph{Statistics}, 23\penalty0 (3):\penalty0 265--270, 1992.

\bibitem[Zeileis and Hothorn(2002)]{LMTEST}
Achim Zeileis and Torsten Hothorn.
\newblock Diagnostic checking in regression relationships.
\newblock \emph{R News}, 2\penalty0 (3):\penalty0 7--10, 2002.
\newblock URL \url{https://CRAN.R-project.org/doc/Rnews/}.

\end{thebibliography}

\end{document}